\documentclass[preprint,number,compress]{elsarticle}
\usepackage[utf8x]{inputenc}
\usepackage{amsmath}
\usepackage{amsfonts}
\usepackage[british]{babel}
\usepackage{accents}
\usepackage{mathrsfs}
\usepackage{tikz}
\usepackage{amssymb}
\usepackage{amsbsy}
\usepackage{amsthm}
\usepackage{color}
\usepackage{eucal}
\usepackage{epic,epsfig}
\usepackage{url}

\usepackage{pifont}
\usepackage{natbib}
\usepackage{geometry}
\usepackage{fleqn}
\usepackage{graphicx}
\usepackage{txfonts}
\usepackage{hyperref}

\DeclareMathAccent{\wtilde}{\mathord}{largesymbols}{"65}
\DeclareMathAccent{\what}{\mathord}{largesymbols}{"62}

\def\m@th{\mathsurround=0pt}
\mathchardef\bracell="0365
\def\upbrall{$\m@th\bracell$}
\def\undertilde#1{\mathop{\vtop{\ialign{##\crcr
    $\hfil\displaystyle{#1}\hfil$\crcr
     \noalign
     {\kern1.5pt\nointerlineskip}
     \upbrall\crcr\noalign{\kern1pt
   }}}}\limits}

\def\m@th{\mathsurround=0pt}
\mathchardef\bracelll="0362
\def\upbralll{$\m@th\bracelll$}
\def\underhat#1{\mathop{\vtop{\ialign{##\crcr
    $\hfil\displaystyle{#1}\hfil$\crcr
     \noalign
     {\kern1.5pt\nointerlineskip}
     \upbralll\crcr\noalign{\kern1pt
   }}}}\limits}

\newcommand{\uh}[1]{ \underhat {#1}}

\newcommand{\ub}[1]{ \underline {#1}}
\newcommand{\oh}[1]{ \what {#1}}
\newcommand{\ot}[1]{ \wtilde {#1}}
\newcommand{\ob}[1]{ \overline {#1}}
\newcommand{\ud}{\mathrm{d}}

\newcommand{\arctanh}{\, \mathrm{arctanh}}

\newcommand{\ox} {\mathbf{x}}
\newcommand{\oI} {\mathbf{I}}
\newcommand{\oX} {\mathbf{X}}

\newcommand{\bB} {\boldsymbol B}
\newcommand{\bn} {\boldsymbol n}

\newcommand{\bu} {\boldsymbol u}

\newcommand{\mca} {\mathfrak a}
\newcommand{\mcD} {\mathcal D}

\newcommand{\mcL} {\mathcal L}
\newcommand{\mcN} {\mathcal N}
\newcommand{\mcS} {\mathcal S}

\newcommand{\ol}[1]{ \ob {#1}}
\newcommand{\wh}[1]{ \oh {#1}}

\newtheorem{prop}{Proposition}
\newdefinition{rmk}{Remark}
\newproof{pf}{Proof}

\title{Quantum Variational Principle and Quantum Multiform Structure: the Case of Quadratic Lagrangians}
\author{S. D. ~King \corref{cor1}\fnref{fn1}}
\author{F. W. ~Nijhoff\fnref{fn2}}
\address{School of Mathematics, University of Leeds, Leeds, United Kingdom LS2 9JT}
\date{\today}

\cortext[cor1]{Corresponding author.}
\fntext[fn1]{Email address: steven.king@cantab.net}
\fntext[fn2]{Email address: F.W.Nijhoff@leeds.ac.uk}

\begin{document}

\begin{abstract}

A modern notion of integrability is that of multidimensional consistency (MDC), 
which classically implies the coexistence of (commuting) dynamical flows in several independent 
variables for one and the same dependent variable. This property holds for both continuous 
dynamical systems as well as for discrete ones defined in discrete space-time. 
Possibly the simplest example in the discrete case is that of a linear quadrilateral lattice 
equation, which can be viewed as a linearised version of the well-known lattice potential 
Korteweg-de Vries (KdV) equation. In spite of the linearity, the MDC property is non-trivial in terms of the 
parameters of the system. The Lagrangian aspects of such equations, and their nonlinear 
analogues, has led to the notion of Lagrangian multiform structures, where the Lagrangians 
are no longer scalar functions (or volume forms) but genuine forms in a multidimensional 
space of independent variables. The variational principle involves variations not only with respect 
to the field variables, but also with respect to the geometry in the space of independent 
variables. In this paper we consider a quantum analogue of this new variational principle 
by means of quantum propagators (or equivalently Feynman path integrals). In the case of quadratic 
Lagrangians these can be evaluated in terms of Gaussian integrals. We study also periodic reductions 
of the lattice leading to discrete multi-time dynamical commuting mappings, the simplest example 
of which is the discrete harmonic oscillator, which surprisingly reveals a rich integrable 
structure behind it. On the basis of this study we propose a new quantum variational principle 
in terms of multiform path integrals.

\end{abstract}


\maketitle


\section{Introduction}


Discrete integrable systems \cite{hietarinta2016discrete} have started to play an increasingly important 
role in deepening the understanding of integrability as a mathematical notion, thereby forging new perspectives 
in both analysis (e.g. the discovery of difference analogues of the Painlev\'e equations), geometry 
(the development of discrete differential geometry, \cite{bobenko2008discrete}) and algebra (e.g. the 
development of cluster algebras through the so-called Laurent phenomenon). In physics, at the quantum 
level, discrete integrable systems appear in connection with random matrix theory and quantum spin 
models of statistical mechanics, and in aspects of relativistic many-body systems \cite{ruijsenaars1986new}, 
but more directly in approaches to establish integrable quantum field theories on the space-time 
lattice \cite{volkov1992quantum}.   

Integrable systems are important not only because they can be treated by exact and rigorous methods, 
but also because they appear to be universal: they have a rare tendency of emerging in a large variety
of contexts and physical situations, such as in correlations functions in scaling limits, random matrices
and in energy level statistics of even chaotic systems. Furthermore, their intricate underlying structures 
gave rise to new mathematical theories, such as quantum groups and cluster algebras, 
revealing novel 
types of combinatorics. Thus, one could argue, letting these systems ``speak for themselves" the 
stories they tell us will lead us to new principles and insights, even perhaps about the structure of Nature 
itself. One such story is about their variational description in terms of a least-action 
principle and its connection to one of the key integrability features, multi-dimensional consistency (MDC). 
The latter is the phenomenon that integrable equations do not come in isolation, but tend to come in 
combination with whole families of equations, all simultaneously imposable on one and the same field 
variable (the dependent variable of the equations). Such equations manifest themselves as higher or 
generalized symmetries, as \textit{hierarchies} of equations or as compatible systems, their very 
compatibility being the signature of the integrability. In fact, it is this very feature that forms a powerful 
tool in the exact solvability of such equations through techniques such as the inverse scattering transform 
(a nonlinear analogue of the Fourier transform), Lax pairs and B\"acklund transformations.  

This story about the variational description of integrable systems started with the paper 
\cite{lobb2009lagrangian}, where the Lagrangian structure of a class of 2D quadrilateral lattice equations was 
studied, which are integrable in the sense of the MDC property. It was shown that for particularly well-chosen 
discrete Lagrangians for those equations, embedded through the MDC property in higher-dimensional space-time 
lattice, the Lagrangians obey a closure property, suggesting that these Lagrangians should be viewed as 
components of a discrete $p$-form that is closed on solutions of the quad equations. This remarkable property 
led to the formulation of a novel least-action principle in which the action is supposed to attain a critical 
point not only w.r.t. variations of the field variables, but also the action being stationary w.r.t. variations 
of the space-time surfaces in the higher-dimensional lattice of independent discrete variables on which the 
equations are defined. This allows one to derive from this extended variational principle not one single 
equation (in the conventional way on a fixed space-time surface) but the full set of compatible equations that 
possess the MDC property. Furthermore, this property was also shown to extend to corresponding integrable 
differential equations defined on smooth surfaces in a multidimensional space-time of independent 
continuous variables, as well as on systems of higher dimension and of higher rank, \cite{lobb2009lagrangianKP,lobb2010lagrangian,xenitidis2011lagrangian} 
as well as to many-body systems \cite{yoo2011discreteCM,yoo2011discreteRS,suris2013variational}. 
Further extensions and deepening understanding of these results were obtained in a number of papers, cf. 
\cite{lobb2013variational,boll2013integrability,boll2013multi}.  

A natural question is whether the Lagrangian multiform structure described above extends also to the quantum 
regime, since, after all, a canonical quantization formalism for reductions of quadrilateral lattice equations 
and higher-rank systems, using non-ultralocal $R$ matrix structures, was already established some while ago 
\cite{nijhoff1993quantization,nijhoff1993fusion}, as well as for a quantum lattice Hirota type system 
\cite{faddeev1999algebraic}, cf. also \cite{bazhanov2008exact}. However, the natural setting for a 
Lagrangian approach in the quantum case is obviously the Feynman path integral \cite{feynman1965quantum}, which has remained curiously unexplored 
in the context of integrable systems theory where there has been a predilection for the Hamiltonian point of view. 
However, when dealing with discrete systems, e.g. systems evolving in discrete time, the Hamiltonian view point is no longer natural, and 
the Lagrangian point of view may become preferable. The further advantage is that in discrete time, path integrals are no longer marred 
by the infinite time-slicing limit which causes such objects to be notoriously ill-defined in general. Thus, first steps to set up a 
path integral approach for integrable quantum mappings\footnote{The notion of \textit{quantum mapping} is essentially due to M.V. Berry et al., 
\cite{berry1979quantum}. }, i.e. integrable systems with discrete-time evolution, were undertaken in 
\cite{field2006time,field2005thesis}. However, the main aim of the present paper is to arrive at an understanding of the Lagrangian multiform 
structure on the quantum level. In order to achieve that, and to avoid analytical complications arising from the nonlinearities, we restrict ourselves 
in this initial treatment to the case of quadratic Lagrangians, associated with linear multidimensionally consistent equations. 
Although this may seem restrictive, the quadratic case is surprisingly rich and exhibits most of the properties of the wider classes of nonlinear 
models when it comes to the MDC aspects. Those reveal themselves in the way the lattice parameters govern the compatible systems of equations, 
and it is there where even these linear equations exhibit quite non-trivial features. In fact, an interesting role reversal between 
discrete independent variables and continuous parameters allows the corresponding quantum propagators to be interpreted at the same time as discrete
as well as continuous path integrals. The periodic reductions are particularly noteworthy, since they lead to propagators that can be readily 
computed, and it is here that the humble quantum harmonic oscillator makes its reappearance in quite a new context.     
 
The outline of the paper is as follows. In section 2 we describe the classical quad equation, i.e. a 2-dimensional partial 
difference equation defined on elementary quadrilaterals, and its Lagrangian 2-form structure. In section 3, we consider its periodic reductions on 
the classical level, and construct commuting flows for the lowest period cases. The simplest (3-step) reduction 
leads to the harmonic oscillator, but even this case there is a non-trivial Lagrangian 1-form structure on the 
classical level. Next, in section 4 we consider the quantization of the reductions through discrete-time step 
path integrals which at the same time provides a natural discretization of the underlying continuous-time model 
in terms of the lattice parameters. The MDC property here is reflected in a path-independence property of the 
propagators. This leads us to suggest a quantum variational principle which we expect may extend to models 
beyond the quadratic case. In section 5 we return to the quad lattice case, which resembles  
a quantum field type of situation, and we establish surface-independence of the relevant propagators, 
suggestive of a quantum variational principle in the field theoretic case.  Finally, in section 6 we discuss 
some possible ramifications of our findings, and how they connect to some ongoing questions 
regarding quantum mechanics and foundational aspects.

\section{Linearised Lattice KdV Equation}
\label{sec:lattice}

Our starting point is a 2 dimensional quadrilateral lattice equation, whose dependent variable $u(n,m)$ is defined on  
lattice points labelled by discrete variables $(n,m)$, which are variables shifting by units, and with lattice parameters 
$p$ and $q$, each  associated with the $n$ and $m$ directions on the lattice respectively. We adopt the shift notation 
by accents $\ot{\phantom{a}}$ and $\oh{\phantom{a}}$, i.e. for $u:=u(n,m)$,we have 
$\ot u := u(n+1,m)$, $\oh u:= u(n,m+1)$. 
The equation of interest in this paper is in the linear quadrilateral equation:
\begin{equation}
   (p+q) (\ot u - \oh u) = (p-q) (u - \oh{\ot u})  \ .
\label{eq:linear}
\end{equation}
This quadrilateral equation is supposed to hold on every elementary
plaquette across a 2 dimensional
lattice; the elementary plaquette is shown in figure \ref{fig:plaquette}.
This is something of a ``universal'' linear quad equation, being the
natural linearisation of nearly all the integrable quad equations of the
ABS list \cite{adler2003classification}.
This equation can be derived
 via discrete Euler-Lagrange equations
on the three-point Lagrangian
\begin{equation}
   \mathcal L (u, \ot u, \oh u) = u (\ot u - \oh u) 
                           - \frac12 \frac{p+q}{p-q} (\ot u - \oh u)^2 
\ ; \quad
   \oh{ \ot{ \left(\frac{\partial \mcL}{\partial u} \right)}} + 
   \oh{ \left(\frac{\partial \mcL}{\partial \widetilde{u}} \right)} + 
   \ot{ \left(\frac{\partial \mcL}{\partial \widehat{u}} \right)} = 0 ,
\label{eq:linearlagrangian}
\end{equation}
where, for the action, we sum across every plaquette in the lattice:
\begin{equation}
    \mathcal S = \sum_{(n,n) \in \mathbb Z^2} 
    \mathcal L(u_{n,m} , u_{n+1,m} , u_{n,m+1}) \ .
\end{equation}
Note  that the Lagrangian (\ref{eq:linearlagrangian})
is also the natural linearisation of the Lagrangians
for the non-linear quad equations of the ABS list from which
(\ref{eq:linear}) can be derived.

\begin{figure}[hbtp]
  \begin{center}
    \begin{tikzpicture}  [scale=2]
      \begin{scope} [very thin]
        \draw [->]  (0,0.5) -- (0,-1.5) node [black][anchor=east] {$m$} ;
        \draw [->]  (-0.5,0) -- (1.5,0) node [black][anchor=south] {$n$} ;
        \draw  (-0.5,-1) -- (1.5,-1) ;
        \draw   (1, 0.5) -- (1,-1.5) ;
      \end{scope}
        \filldraw [fill=black!5!white , draw =black, thick] 
              (0,0) --  (1,0) -- (1,-1) -- (0,-1)  -- cycle  ;
        \filldraw [fill=white, thick] circle (0.04) 
              node[anchor=south east] {$u$} ;
        \filldraw [fill=white, thick] (0,-1) circle (0.04) 
              node[anchor=north east] {$\what u$} ;    
        \filldraw [fill=white, thick] (1,0) circle (0.04)
              node[anchor = south west] {$\wtilde u$} ;
        \filldraw [fill=white, thick] (1,-1) circle (0.04) 
              node[anchor = north west] {$ \what{\wtilde u}$} ;
        \draw (0.5,-0.5) node {$\mathcal L (u, \wtilde u, \what u)$}  ;
        \draw (0.5,0) node [anchor=south] {$p$} ;
        \draw (0,-0.5) node [anchor=east] {$q$} ;   
    \end{tikzpicture}
\caption{An elementary plaquette in the lattice}
\label{fig:plaquette}
\end{center} 
\end{figure}
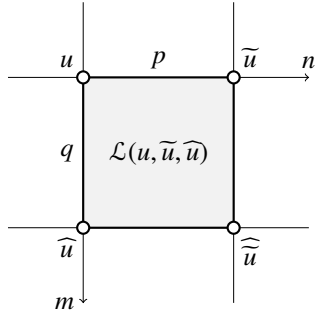

In fact, the standard variational principle on (\ref{eq:linearlagrangian}) produces
two copies of (\ref{eq:linear}). In order to regain precisely the linearised KdV 
equation, we must make use of the multiform variational principle introduced by Lobb
and Nijhoff \cite{lobb2009lagrangian,lobb2013variational}. (\ref{eq:linear})
can be consistently embedded into a \emph{multidimensional} lattice,
with directions labelled by subscripts $i,j,k$. Across an elementary plaquette
in the $i-j$ plane, (\ref{eq:linear}) takes the form:
\begin{equation}
   (p_i + p_j) (u_i - u_j) = (p_i - p_j) (u - u_{ij})  \ ,
\label{eq:generallinear}
\end{equation}
where $u_i$ indicated $u$ shifted once in the $i$ direction on the lattice, 
and $p_i$ is now the lattice parameter associated to the $i$ direction.
This equation has multidimensional consistency, which can be checked by establishing
closure around the cube \cite{nijhoff2001discrete} - field variables at any
point in the multi-dimensional lattice can be calculated via any route
in a consistent manner.

In the variational principle proposed in \cite{lobb2013variational}, the action
is defined as the sum of Lagrangians on elementary plaquettes
 across a 2-dimensional surface $\sigma$, embedded in the multidimensional space.
To derive the equations of motion, we then demand the action be stationary
not only under the variation of the field variables $u$, but also under the 
variation of the surface $\sigma$ itself. For this to hold, we require
closure of the Lagrangian:  if we consider the combination of oriented Lagrangians on the faces of a cube,
we require that \emph{on the equations of motion}, the Lagrangians sum to zero.
In other words,
\begin{equation}
     \Delta_1  \mathcal L_{23} (u) + \Delta_2 \mathcal L_{31} (u) 
    + \Delta_3 \mathcal L_{12} (u) 
     =  \mathcal L_{23} (u_1) - \mathcal L_{23} (u) 
    + \mathcal L_{31} (u_2) - \mathcal L_{31} (u) 
    + \mathcal L_{12} (u_3) - \mathcal L_{12} (u)
    = 0    \ ,
\label{eq:closure}
\end{equation}
where we have used the shorthand 
$\mcL_{ij} (u) := \mcL (u, u_i, u_j; p_i, p_j)$, and
the final equality in (\ref{eq:closure}) holds only when we apply (\ref{eq:generallinear}).
According to \cite{lobb2013variational}, such a system must be described
by a Lagrangian of the form
$
    \mathcal L (u, u_i , u_j ; p_i,p_j) = A (u, u_i;p_i) - A(u,u_j;p_j)
                                                    + C(u_i, u_j; p_i,p_j) \ ;
$
where we require $C_{ij}$ to be antisymmetric under interchange of $i$ and $j$.
Notice that the Lagrangian (\ref{eq:linearlagrangian}) is already in this form.
By using the  multidimensional consistency, a set of Euler-Lagrange equations
are derived, which simplify on a single plaquette to:
\begin{equation}
    \frac {\partial}{\partial u_i} \Big(
    A(u, u_i; p_i) - A(u_i, u_{ij}; p_j) + C(u_i,u_j;p_i,p_j)
    \Big)
    = 0 \ .
\label{eq:latticeEL}
\end{equation}
This yields precisely the equation (\ref{eq:linear}).
This structure allows us to describe the mutliple consistent equations (\ref{eq:generallinear})
in a single Lagrangian framework - that of the 2-form. This is then the appropriate
variational structure to describe multi-dimensionally consistent systems \cite{lobb2009lagrangian}.


In fact, the Lagrangian (\ref{eq:linearlagrangian}) is the almost unique quadratic
Lagrangian with a 2-form structure (i.e. exhibiting the closure property).
Considering the general form for a three-point Lagrangian 2-form
and equation of motion (\ref{eq:latticeEL}), 
we restrict our attention to quadratic Lagrangians and have the 
general form:
\begin{equation}
    \mcL_{ij} (u, u_i, u_j) = \left( \tfrac12 a_i u^2 + c_i u u_i \right)
    - \left( \tfrac12 a_j u^2 + c_j u u_j \right)
    + \left( \tfrac12 b_{ij} u_i^2 - \tfrac12 b_{ji} u_j^2 + \delta_{ij} u_i u_j
    \right)  \ ,
\label{eq:genclasslagrange}
\end{equation}
where we require $\delta_{ji} = -\delta_{ij}$.
Here, subscripts on coefficients indicate dependence on the lattice parameters
$p_i$ and $p_j$.
This Lagrangian yields the equation of motion:
$
    c_i u - c_j u_{ij} = (a_j - b_{ij})u_i - \delta_{ij} u_j
$.
This is a quad equation, and as such we require it to be symmetric under the 
interchange of $i$ and $j$. This leads to the conditions
$
    c_i = c_j = c \ , \ \textrm{constant,}
\ 
    a_j - b_{ij} = \delta_{ij} \ .
$

Noting that the Lagrangian (\ref{eq:genclasslagrange}) already obeys the closure
relation (\ref{eq:closure}) on the equations of motion above,
we use our freedom to multiply by an overall constant to let $c=1$, and hence
the general Lagrangian is given by:
\begin{equation}
    \mcL_{ij} (u, u_i, u_j) = 
    u (u_i - u_j) - \tfrac12 \delta_{ij} (u_i - u_j)^2
    + \tfrac12 a_i (u^2 - u_j^2) - \tfrac12 a_j (u^2 - u_i^2) \ .
\label{eq:linearlagrangian2}
\end{equation}
We can see this has the same form as (\ref{eq:linearlagrangian}), but with
a more general dynamical, anti-symmteric parameter $\delta_{ij}$, 
and the free parameter
$a_i$ that does not effect the equations of motion.


\section{One Dimensional Reduction: The Discrete Harmonic Oscillator}

\subsection{Periodic Reduction}
\label{sec:periodicreduction}

Reductions of lattice equations to integrable symplectic mappings have been considered since the early 1990s \cite{papageorgiou1990integrable,bruschi1991integrable,capel1991complete,quispel1991integrable}. Here, we are considering a 
linearised version of the lattice KdV equation as our starting point, and follow the same reduction procedure
as has been considered previously for non-linear quad equations. 
The reduction is obtained by imposing a periodic initial value problem, where the evolution of the data progresses through the lattice
according to a dynamical map, or equivalently a system of ordinary difference equations, which is constructed by implementing the lattice 
equation (\ref{eq:linear}).
We begin with initial data $u_0$, $u_1$ and $u_2$, and let $\oh u_2 = u_0$, according
to figure \ref{fig:reduction}. This unit is then repeated periodically
across an infinite staircase in the lattice.
This is the simplest meaningful reduction we can perform on the lattice
equation.

\begin{figure}[hbtp]
\begin{center}
    \begin{tikzpicture}  [scale=2]
      \draw [thick] (0,0.2) -- (0,0) -- (2,0) -- (2,-1) -- (2.2,-1) ;
      \draw [thick, dashed]  (0,0) -- (0,-1) -- (2,-1)  (1,0) -- (1,-1) ;
              
      \filldraw [fill=black] circle (0.04) 
          node[anchor=south west] {$u_0$} ;
      \filldraw [fill=black] (1,0) circle (0.04)
          node[anchor = south west] {$u_1$} ;
      \filldraw [fill=black] (2,0) circle (0.04) 
          node[anchor = south west] {$ u_2$} ;
      \filldraw [fill=black] (2,-1) circle (0.04) 
          node[anchor = south west] {$ \oh u_2 = u_0$} ; 
                    
      \filldraw [fill=white, thick] (0,-1) circle (0.04) 
          node[anchor=south west] {$\oh u_0$} ;    
      \filldraw [fill=white, thick] (1,-1) circle (0.04) 
          node[anchor=south west] {$\oh u_1$} ;   

      \draw (0.5,0) node [anchor=south] {$p$} ;
      \draw (0,-0.5) node [anchor=east] {$q$} ;   
    \end{tikzpicture}
\caption{Periodic inivital value problem on the lattice equation}
\label{fig:reduction}
\end{center}
\end{figure}
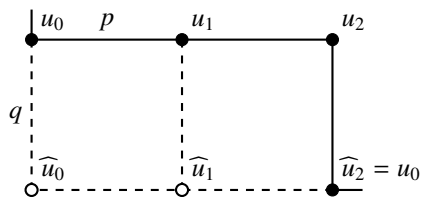

Applying the linear lattice equation (\ref{eq:linear}) to each
plaquette, we can write equations for the dynamical mapping
$ (u_0,u_1,u_2) \to (\oh u_0, \oh u_1, \oh u_2) $:
\begin{equation}
  \oh{u}_0 = u_1+s(\oh{u}_1-\oh{u}_2)
\  , \quad
  \oh{u}_1 = u_2+s(u_0-u_1)
\  , \quad  
  \oh{u}_2 = u_0
\ ; \quad
  s:=\frac{p-q}{p+q} \ . 
\label{eq:hatumap}
\end{equation}
This is a finite-dimensional discrete system. We introduce the
reduced variables
$
    x := u_1 - u_0 $ , $  y := u_2 - u_1
$
and, by eliminating $y$, write the second order 
difference equation:
\begin{equation}
  \oh{x} + 2bx + \uh x = 0\  , \quad b:=1+2s-s^2\ ,
\label{eq:p=1.5xb}
\end{equation}
where the underhat $\uh x$ indicates a backwards step.
This equation can be expressed by a Lagrangian-type generating function,
with the equation arising from discrete Euler-Lagrange equations:
\begin{equation}
    \mathcal L (x, \oh x) = x \oh x + b x^2
\ , \quad 
    \oh{\frac{\partial \mcL}{\partial x}} + \frac{\partial \mcL}{\partial \oh{x}} = 0 \ ,
\label{eq:simplelagrangian}
\end{equation}
and so is symplectic,
$
    \ud \oh x \wedge \ud \oh y = \ud x \wedge \ud y  .
$
The map also possesses an exact invariant:
\begin{equation}
    I_b (x ,\oh x) = x^2 + \oh x^2 + 2b x \oh x \ .
\label{eq:invariantb}
\end{equation}

The equation (\ref{eq:p=1.5xb}) is a discrete harmonic oscillator. 
It is not difficult to see that the most general solution to (\ref{eq:p=1.5xb}) 
is given by
\begin{equation}
    x_m = c_1 \sin (\mu m) + c_2 \cos (\mu m) 
\ ; \quad
    \cos \mu = -b \ , 
\label{eq:hosolution}
\end{equation}
where $m$ is the discrete variable. This has a clear relation to the solution
for the continuous time harmonic oscillator.
This solution can alternatively be written as
$
    x_m = A \lambda^m + B \lambda^{-m}
$, $
    \lambda = -b + \sqrt{b^2-1}
$.
By considering derivatives with respect to the \emph{parameter} $b$, we
can then derive the equations:
\begin{equation}
    \frac{d x}{d b}
    = \frac m{1-b^2} (b x + \oh x)
\ , \ 
    \frac{d x}{d b}
    = - \frac m{1-b^2} (b x + \uh x) \ , 
\label{eq:semidisc} 
\end{equation}
Eliminating $\oh x$ yields the second order differential equation in $b$:
\begin{equation}
    (1-b^2) \frac{d^2 x}{d b^2}
    - b \frac{d x}{d b} + m^2 x = 0
\ .
\label{eq:bcontflow}
\end{equation}
 A remarkable exchange has taken place: the parameter
and independent variable of the discrete case, $b$ and $m$, have exchanged
roles to become the independent variable and parameter of a continuous time
model.
Note that (\ref{eq:bcontflow}) can be simplified by taking $\mu:= \cos^{-1}(-b)$ as the
 ``time'' variable, so that:
$
    d^2 x/ d\mu^2 + m^2 x = 0
$.
This is the equation for the harmonic oscillator, with a quantised
frequency $\omega = m$.


\subsection{Commuting Discrete Flow}
\label{sec:commdiscflow}

Recall that the linear lattice equation (\ref{eq:generallinear})
can be embedded in a multidimensional lattice. 
From the
periodic reduction in the plane (figure \ref{fig:reduction})
we consider the embedding within a three dimensional lattice.
The third lattice direction has lattice parameter $r$, and we
introduce shifted variables $\ob u_i$, as shown in figure \ref{fig:commuting}.

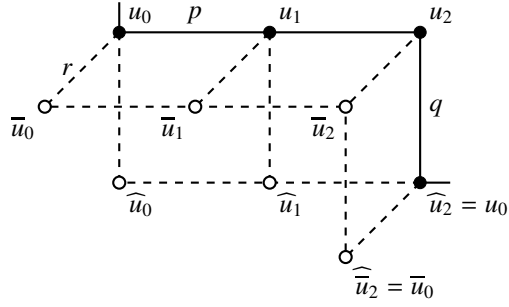
\begin{figure}[hbtp]
\begin{center}
   \begin{tikzpicture}  [scale=2]
    \draw [thick] (0,0.2) -- (0,0) -- (2,0) -- (2,-1) -- (2.2,-1) ;
    \draw [thick, dashed]  (0,0) -- (0,-1) -- (2,-1)  (1,0) -- (1,-1) ;
    \filldraw [fill=black] circle (0.04) 
          node[anchor=south west] {$u_0$} ;
    \filldraw [fill=black] (1,0) circle (0.04)
          node[anchor = south west] {$u_1$} ;
    \filldraw [fill=black] (2,0) circle (0.04) 
          node[anchor = south west] {$ u_2$} ;
    \filldraw [fill=black] (2,-1) circle (0.04) 
          node[anchor = north west] {$ \oh u_2 = u_0$} ; 
    \filldraw [fill=white, thick] (0,-1) circle (0.04) 
          node[anchor=north west] {$\oh u_0$} ;    
    \filldraw [fill=white, thick] (1,-1) circle (0.04) 
          node[anchor=north west] {$\oh u_1$} ;   
    \draw (0.5,0) node [anchor=south] {$p$} ;
    \draw (2,-0.5) node [anchor=west] {$q$} ;   
     \draw [thick, dashed] (0,0) -- (225:0.7) -- ++(2,0) -- (2,0) 
                             (1,0) -- +(225:0.7) 
                           (2,-1) -- ++(225:0.7) -- ++(0,1);
    \filldraw [fill=white, thick] (225:0.7) circle (0.04) 
          node[anchor=north east] {$\ob u_0$} ;    
    \filldraw [fill=white, thick] (1,0) +(225:0.7) circle (0.04) 
          node[anchor=north east] {$\ob u_1$} ;   
    \filldraw [fill=white, thick] (2,0) +(225:0.7) circle (0.04) 
          node[anchor=north east] {$\ob u_2$} ; 
    \filldraw [fill=white, thick] (2,-1) +(225:0.7) circle (0.04) 
          node[anchor=north west] {$\oh {\ob u}_2 = \ob u_0$} ; 
    \draw (225:0.35) node [anchor=east] {$r$} ;
  \end{tikzpicture} 
\caption{The variables $\ob u_i$ extend from the plane in a third
direction.}
\label{fig:commuting}
\end{center}
\end{figure}

To derive the mapping, we now use the lattice equations (\ref{eq:generallinear}):
\[ 
    (q+r)(\oh{u}-\ob{u})=(q-r)(u-\oh{\ob{u}})\  , \quad  
    (r+p)(\ob{u}-\ot{u})=(r-p)(u-\ot{\ob{u}})\  , 
\]  
which, in terms of the $u_i$, yield 
\begin{equation}
    \begin{array}{ccc}
        \ob u_0 &=& u_1+t(\ob u_1-u_0) \ , 
    \\ 
        \ob{u}_1 &=& u_2+t(\ob{u}_2-u_1) \ , 
    \\ 
        \ob{u}_2 &=& u_0+t'(\ob{u}_0-u_2) \ . 
    \end{array}
\quad \textrm{where \ }
    \begin{array}{c}
        t := \frac{p-r}{p+r} \  , 
    \\ 
        t' := \frac{q-r}{q+r} \ .
    \end{array}
\end{equation}
Again, we use reduction variables $(x,y)$, which yield
the map $(x,y) \to (\ob x, \ob y)$.
This map can be written in a matrix form,
from which it can be shown to be area preserving, 
$\ud \ob x \wedge \ud \ob y = \ud x \wedge \ud y$.
Eliminating $y$ again produces a second order 
difference equation in $x$:
\begin{equation}
    \ob{x}+2ax+\ub{x}=0\  , 
    \quad {\rm with}\quad 
    2a:=\frac{(2t+1-t^2)-t'(2t-1+t^2)}{1-t^2t'}\  .  
\label{eq:p=1.5xa}
\end{equation}
This equation has the same form as (\ref{eq:p=1.5xb}), that of a discrete harmonic 
oscillator, along with invariant
$
    I_a (x ,\ob x) = x^2 + \ob x^2 + 2a x \ob x
$.

We can write both maps $(x,y) \to (\oh x, \oh y)$ and $(x,y) \to (\ob x, \ob y)$
in matrix form:
$
    \oh{\mathbf x} = \mathbf {S \,x } 
$ , $
    \ob{\mathbf x} = \mathbf{T \, x} 
$ , $
    \mathbf x := ( x , y )^T
$.
It is then clear that the two maps commute, 
$
    (\oh {\ob x} , \oh {\ob y}) = (\ob {\oh x} , \ob {\oh y}) \ ,
$
since we have
$
    [ \mathbf S , \mathbf T ] = 0 \ .
$
This last relation relies on the parameter identity,
$
    s t t' = s - t + t' 
$,
which is easily shown using the definitions for $s$, $t$ and $t'$.

Our equations are slightly simplified by introducing the parameters
$P := p^2 + pq \  $, $Q := q^2 \ $ and $R := r^2 \ $,  in terms of which
$
     a = (P - R)/(P + R)
$ , $
     b = (P - Q)/(P + Q)
$.
By returning to earlier evolution equations in terms of $x$ and $y$ and eliminating $y$
in a different manner, we derive ``corner equations'' for the evolution, 
linking $x$, $\oh x$ and
$\ob x$; or $\oh x$, $\ob x$ and $\oh {\ob x}$ respectively. Thus:
\begin{equation}
    \left( \frac{P-Q}q - \frac{P-R}r \right) x 
                 = \frac{P+R}r \ob x - \frac{P+Q}q \oh x  
\ , \quad
    \left( \frac{P-Q}q - \frac{P-R}r \right) \oh {\ob x} 
                 = \frac{P+R}r \oh x - \frac{P+Q}q \ob x \ .
\label{eq:p=1.5corner}
\end{equation}
Thus we have multiple equations of motion (\ref{eq:p=1.5xb}),
(\ref{eq:p=1.5xa}), (\ref{eq:p=1.5corner}) all holding simultaneously
on the same variable $x$.

\subsection{Lagrangian 1-form structure}
\label{sec:1form}

A recent development in understanding discrete integrable systems 
with commuting flows has been the Lagrangian multiform theory
 \cite{lobb2013variational,lobb2009lagrangian, 
yoo2011discreteCM, yoo2011discreteRS, nijhoff2013lagrangian,  
suris2013variational,boll2013multi}. A system with two or more 
commuting, discrete flows can be
described by a Lagrangian 1-form structure,
which provides a way to obtain a simultaneous system of equations
for a single dependent variable from a variational principle.
Thus, the Lagrangians generating the 
flows $x \to \oh x$ and $x \to \ob x$ should
form the components of a \emph{difference 1-form}, each associated with an 
oriented direction on a 2D lattice.

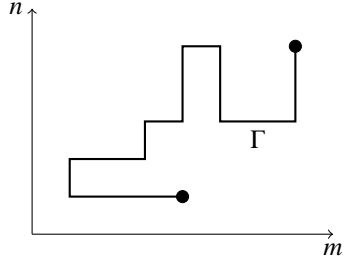
\begin{figure}[htp]
\begin{center}
\begin{tikzpicture}  [scale=1]
    \draw [->] (0.5,0.5) -- (4.5,0.5)   node [anchor=north] {$m$}   ;
    \draw [->] (0.5,0.5) -- (0.5,3.5)   node [anchor=east] {$n$} ;
    \draw [thick] (2.5,1) -- (1,1) -- (1,1.5) -- (2,1.5)
             -- (2,2) -- (2.5,2) -- (2.5,3) -- (3,3)
             -- (3,2) -- (4,2) -- (4,3) ;
    \draw (3.5,2) node[anchor=north] {$\Gamma$} ;
    \filldraw [fill=black] (2.5,1) circle (0.08)  ;
    \filldraw [fill=black] (4,3) circle (0.08)  ;
\end{tikzpicture}
\caption{A curve $\Gamma$ in the discrete variables.}
\label{fig:gamma}
\end{center}
\end{figure}

The action functional is then defined as a sum of elementary Lagrangian elements over 
an arbitrary discrete curve $\Gamma$ in the 2D lattice, as shown in figure 
\ref{fig:gamma}.
\begin{equation}
    \mcS [x(\mathbf n) ; \Gamma] = 
    \sum\limits_{\mathbf {\gamma}(\mathbf n) \in \Gamma }
    \mathcal{L}_i (x(\mathbf n) , x(\mathbf n +\mathbf e_i )) \ .
\label{eq:SGamma}
\end{equation}
The usual variational principle demands that, on the equations of motion,
the action $\mcS$ be stationary under the variation of the dynamical
variables $x$. In addition, we also demand that $\mcS$ be stationary
under variations of the curve $\Gamma$ itself.
 This principle leads to the compatability
of equations of motion and corner equations, under the condition of \emph{closure} of 
the Lagrangians. That is, on the equations of the motion, the action should be locally
invariant under changes to the curve $\Gamma$ and therefore:
\begin{equation}
     \Box \mathcal L := \mcL_a (\oh x, \oh {\ob x})
                             - \mcL_a (x, \ob x)
                             - \mcL_b (\ob x, \oh {\ob x})
                             + \mcL_b (x, \oh x)
                             = 0 \ ,
\label{eq:1formclosure}
\end{equation}
where this last equality holds only on the equations of motion.

In the model we are considering, we already have compatible flows with consistent
corner equations, and so it is natural for us to seek a Lagrangian form exhibiting closure.
However, if we naively seek to satisfy the closure relation (\ref{eq:1formclosure})
with any simple Lagrangian yielding the equations of motion,
we will find that this does not suffice - we must seek a more specific form.
By considering the general form for the quadratic Lagrangians:
\begin{equation}
     \mathcal L_a = 
      \alpha \big (x \ob x + (a - a_0) x^2 + a_0 \ob x^2 \big) 
\ , \quad
     \mathcal L_b =  
     \beta \big( x \oh x + (b - b_0) x^2 + b_0 \oh x^2 \big) \ ,
\label{eq:generaloscillatorlagrangian}
\end{equation}
we can apply the closure $\Box \mathcal L = 0$ as a condition.
Recall that we require closure only on the solutions to the equations of motion,
so we apply the corner equations
(\ref{eq:p=1.5corner}) to $\Box \mcL$, and then compare coefficients
of the remaining terms. 
Demanding that $\alpha, a_0$ and $\beta, b_0$
be independent of $Q$ and $R$ respectively, we find the conditions on the coefficients:
\begin{equation}
     \alpha = \frac{P+R}r \gamma 
\ , \quad
     \beta  = \frac{P+Q}q \gamma 
\ , \quad
     a_0  = \frac r{P+R} f(P) + \frac12 a 
\ ,  \quad
     b_0  = \frac q{P+Q} f(P) + \frac12 b \ .
\label{eq:closureconditions}
\end{equation}
where $\gamma$ is some overall constant, and $f(P)$ is a free function of $P$.
$f$ does not make any contribution to what follows, and so we ignore it:
we let $a_0 = a/2$ and $b_0 = b/2$.

This yields the Lagrangians:
\begin{equation}
     \mathcal L_a (x, \ob x) = \frac1r \bigg( (P+R) \, x \ob x 
                                    + \frac12 (P-R) \, (x^2 + \ob x^2) \bigg) 
\ , \ \ 
     \mathcal L_b (x, \oh x) = \frac1q \bigg( (P+Q) \, x \oh x 
                                    + \frac12 (P-Q) \, (x^2 + \oh x^2) \bigg) .
\label{eq:p=1.5lagrangians}
\end{equation}
By construction, these obey the condition $\Box \mathcal L = 0$
on the equations of motion, and also yield the
equations of motion (\ref{eq:p=1.5xb}) and (\ref{eq:p=1.5xa}) by the usual
variational principle. This eliminates a great deal of the usual freedom in
choosing our Lagrangian: the closure condition mandates a specific form
of the Lagrangian.

In fact, not only the equations (\ref{eq:p=1.5xb}) and (\ref{eq:p=1.5xa})
arise from a variational principle on this action, but also the 
corner equations (\ref{eq:p=1.5corner}).
We have four elementary curves in the space of two discrete variables,
shown in figure \ref{fig:elementarycurves}\footnote{Such elementary curves defining a complete set of discrete EL equations 
were first considered in \cite{yoo2011calogero}.}. Across each curve, we can define an
action, and then a variation with respect to the middle point, which leads to
an equation of motion.

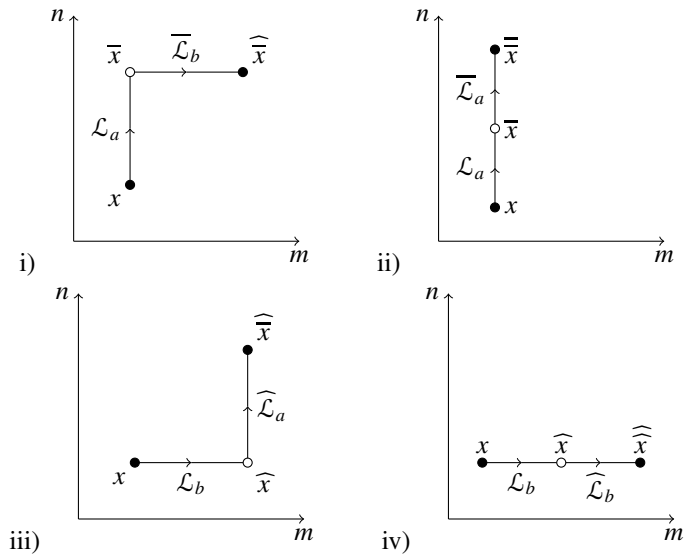
\begin{figure}[htp]
\begin{center}
i)
\begin{tikzpicture}  [scale=1.5]
    \draw [->] (0,0) -- (2,0)   node [anchor=north] {$m$}   ;
    \draw [->] (0,0) -- (0,2)   node [anchor=east] {$n$} ;

    \draw [->] (0.5,0.5) -- (0.5,1) ;
    \draw [->] (0.5,1) -- (0.5,1.5) -- (1,1.5) ;
    \draw (1,1.5) -- (1.5,1.5) ;

    \filldraw [fill=black] (0.5,0.5) circle (0.04) 
              node[anchor=north east] {$x$};
    \filldraw [fill=black] (1.5,1.5) circle (0.04) 
              node[anchor=south west] {$\oh {\ob x}$} ;
    \filldraw [fill=white] (0.5,1.5) circle (0.04)
              node[anchor=south east] {$\ob x$} ;
    
    \draw (0.5,1) node[anchor=east]
           {$\mcL_a$} ;
    \draw (1,1.5) node[anchor=south]
           {$\ob {\mcL}_b$} ;
\end{tikzpicture}
\qquad 
%
ii)
\begin{tikzpicture}  [scale=1.5]
    \draw [->] (0,0) -- (2,0)   node [anchor=north] {$m$}   ;
    \draw [->] (0,0) -- (0,2)   node [anchor=east] {$n$} ;

    \draw [->] (0.5,0.3) -- (0.5,0.65) ;
    \draw [->] (0.5,0.65) -- (0.5,1.35) ;
    \draw (0.5,1.35) -- (0.5,1.7) ;

    \filldraw [fill=black] (0.5,0.3) circle (0.04) 
              node[anchor=west] {$x$};
    \filldraw [fill=white] (0.5,1) circle (0.04) 
              node[anchor=west] {$\ob x$} ;
    \filldraw [fill=black] (0.5,1.7) circle (0.04) 
              node[anchor=west] {$\ob {\ob x}$} ;
    
    \draw (0.5,0.65) node[anchor=east]
           {$\mcL_a$} ;
    \draw (0.5,1.35) node[anchor=east]
           {$\ob {\mcL}_a$} ;
\end{tikzpicture}
\qquad 
\\
iii)
\begin{tikzpicture}  [scale=1.5]
    \draw [->] (0,0) -- (2,0)   node [anchor=north] {$m$}   ;
    \draw [->] (0,0) -- (0,2)   node [anchor=east] {$n$} ;

    \draw [->] (0.5,0.5) -- (1,0.5) ;
    \draw [->] (1,0.5) -- (1.5,0.5) -- (1.5,1) ;
    \draw (1.5,1) -- (1.5,1.5) ;
    
    \filldraw [fill=black] (0.5,0.5) circle (0.04) 
              node[anchor=north east] {$x$};
    \filldraw [fill=white] (1.5,0.5) circle (0.04) 
              node[anchor=north west] {$\oh x$} ;
    \filldraw [fill=black] (1.5,1.5) circle (0.04) 
              node[anchor=south west] {$\oh {\ob x}$} ;
    
    \draw (1,0.5) node[anchor=north]
           {$\mcL_b $} ;
    \draw (1.5,1) node[anchor=west]
           {$\oh {\mcL}_a $} ;
\end{tikzpicture}
\qquad 
%
iv)
\begin{tikzpicture}  [scale=1.5]
    \draw [->] (0,0) -- (2,0)   node [anchor=north] {$m$}   ;
    \draw [->] (0,0) -- (0,2)   node [anchor=east] {$n$} ;

    \draw [->] (0.3,0.5) -- (0.65,0.5) ;
    \draw [->] (0.65,0.5) -- (1.35,0.5) ;
    \draw (1.35,0.5) -- (1.7,0.5) ;

    \filldraw [fill=black] (0.3,0.5) circle (0.04) 
              node[anchor=south] {$x$};
    \filldraw [fill=white] (1,0.5) circle (0.04) 
              node[anchor=south] {$\oh x$} ;
    \filldraw [fill=black] (1.7,0.5) circle (0.04) 
              node[anchor=south] {$\oh {\oh x}$} ;
    
    \draw (0.65,0.5) node[anchor=north]
           {$\mcL_b$} ;
    \draw (1.35,0.5) node[anchor=north]
           {$\oh {\mcL}_b$} ;
\end{tikzpicture}
\end{center}
\caption{Simple discrete curves for variables $m$ and $n$.}
\label{fig:elementarycurves}
\end{figure}

The action and Euler Lagrance equation for curve \ref{fig:elementarycurves}(i) are
\begin{equation}
    S=\mathcal{L}_a(x,\ol{x})+\mathcal{L}_b(\ol{x},\wh{\ol{x}})
\ , \quad
    \frac{\partial S}{\partial\ol{x}}=
    2\left[\left(\tfrac{P-R}{r}+\tfrac{P-Q}{q}\right)\ol{x}+\tfrac{P+R}{r}x+\tfrac{P+Q}    
    {q}\wh{\ol{x}}\right]=0\  ,
\end{equation}
which is compatible with equations (\ref{eq:p=1.5corner}).
Similarly, for curve \ref{fig:elementarycurves}(ii):
\begin{equation}
    S=\mathcal{L}_a(x,\ol{x})+\mathcal{L}_a(\ol{x},\ol{\ol{x}})
\  ,  \quad
    \frac{\partial S}{\partial\ol{x}}=2\left[ 2\tfrac{P-R}{r}\ol{x}+
    \tfrac{P+R}{r}\left(x+\ol{\ol{x}}\right)\right]=0\  ,  
\end{equation}
which is equation (\ref{eq:p=1.5xa}) (i.e. this is a ``standard" Euler-Lagrange equation).
Curves \ref{fig:elementarycurves}(iii) and (iv) yield similarly (\ref{eq:p=1.5xb}) and the 
other part of  (\ref{eq:p=1.5corner}).
We therefore have, \emph{for the specific choice of Lagrangians described}, a
consistent 1-form structure, yielding the equations of motion and corner equations,
and obeying a Lagrangian closure relation.
The discrete harmonic oscillator then, despite its simplicity, nonetheless
has an underlying structure of a Lagrangian one-form expressing
commuting flows: this is the simplest example yet discovered of
such a structure.

Recall the invariants,
it is straightforward to show using the equations of motion that both invariants 
are preserved under both evolutions,
$
     \oh I_b = \ob I_b = I_b $ , $\  \ob I_a = \oh I_a = I_a 
$.
It is not clear, however, that these invariants are necessarily equal: $I_b$ has an
apparent dependence on $Q$, and $I_a$ on $R$, that must be resolved.
Taking our special choice of Lagrangians (\ref{eq:p=1.5lagrangians}),
we can then define canonical momenta, and rewrite our invariants in those terms.
Writing $X_a$ as the mometum conjugate to $x$ in $\mathcal L_a$, and $X_b$ 
similarly for $\mathcal L_b$, we find:
\begin{equation}
     X_a = -\frac {\partial \mathcal L_a}{\partial x} \ = \ 
                  - \frac{P+R}r \, \ob x \ - \frac{P-R}r \, x 
\ , \quad
     X_b = -\frac {\partial \mathcal L_b}{\partial x} \ = \ 
                  - \frac{P+Q}q \, \oh x \ - \frac{P-Q}q \, x \ .
\label{eq:momenta}
\end{equation}
As a direct consequence of the corner equation (\ref{eq:p=1.5corner}) we 
then have precisely that
$
     X_a = X_b =: X \ .
$
In other words, we can define a common conjugate momentum for both evolutions.
If we then write our invariants in terms of $x$ and $X$ we find after multiplication
by a constant (which clearly does not change the nature of the invariants) that
\begin{equation}
    I_a = I_b = \frac12 X^2 + 2P x^2 \ .
\label{eq:p=1.5invariant}
\end{equation}
Note that in this form $I_a, I_b$ appear $Q$ and $R$ independent, and are nothing other than
 the Hamiltonian for the continuous harmonic oscillator, with angular frequency 
 $\omega = 2 \sqrt P$.
This form is Lagrangian dependent.
A different choice of Lagrangian yields different conjugate
momenta that are no longer equal, and where the equivalence of the
invariants is no longer apparent. Requiring equality of the invariants
turns out to be an equivalent condition to demanding Lagrangian closure.

The compatibility of the two discrete evolutions and their corner equations (guaranteed 
by the Lagrangian 1-form structure) allows us to consider a joint solution
to the equations $x_{m,n}$. We allow $m$ to label the hat evolution, and $n$ to label
the bar evolution, such that $x=x_{m,n}$, $\oh x = x_{m+1,n}$, 
$\ob x = x_{m,n+1}$, and so on.
Requiring $x_{m,n}$ to obey (\ref{eq:p=1.5xb}), (\ref{eq:p=1.5xa})
and (\ref{eq:p=1.5corner}), we have the joint solution for the
evolutions:
\begin{equation}
    x_{m,n} = c_1 \sin (\mu m + \nu n) + c_2 \cos (\mu m + \nu n) 
\ ; \quad
    b = - \cos \mu
\ , \
    a = -\cos \nu \ .
\label{eq:jointsoln}
\end{equation}
In the same way as the parameter $b$ generates a continuous flow compatible
with the discrete evolution (\ref{eq:bcontflow}), so we can find a continuous
flow in the parameter $a$:
\begin{equation}
    (1-a^2) \frac{d^2 x}{d a^2}
    - a \frac{d x}{d a} + n^2 x = 0 \ .
\label{eq:contflows}
\end{equation}
Now the joint solution (\ref{eq:jointsoln}) guarantees the compatibility of the
$a$ and $b$ flows with the commuting discrete evolutions.
The compatibility of the continuous flows can be further verified by checking the relation
$
 \tfrac d{da} \tfrac{dx}{db} = \tfrac d{db} \tfrac{dx}{da}
$ 
using (\ref{eq:semidisc}) and similar equations for $a$.
The continuous time-flows are generated by 
the usual Euler-Lagrange equations on continuous time Lagrangians of the form
\begin{equation}
    \mcL_b (x, x_b) = \frac1{2m} \sqrt{1-b^2} \  \left(\frac{\partial x}{\partial b}\right)^2
    - \frac{m}{2\sqrt{1-b^2}} x^2 
\ ; \quad
    \mcL_a (x, x_a) = \frac1{2n} \sqrt{1-a^2} \  \left(\frac{\partial x}{\partial a}\right)^2
    - \frac{n}{2\sqrt{1-a^2}} x^2  .
\end{equation}
Using the corner equations (\ref{eq:p=1.5corner})
these Lagrangians exhibit \emph{continuous}
multiform compatibility, obeying the relations
\begin{equation}
    \frac{\partial \mcL_a}{\partial x_a}
    = \frac{\partial \mcL_b}{\partial x_b}
\ , \quad
    \frac{\partial}{\partial a} 
    \left( \frac{\partial \mcL_b}{\partial x} \right)
    = \frac{\partial}{\partial b} 
    \left( \frac{\partial \mcL_a}{\partial x} \right)
\ .
\end{equation}
So, by considering the discrete parameters $a,b$ now as
\emph{continuous variables}, we find a continuous-time
1-form structure.

As in \cite{degasperis2001newton}, the harmonic oscillator continues
to display surprising new features. On the discrete level,
we discover compatible flows that can be expressed through
the structure of a Lagrangian form, even for this very simple
case. 
This deeper structure then extends
beyond the discrete case also into compatible continuous flows
and we have an interplay between these discrete and continuous
one-form structures.
Having endowed the harmonic oscillator with these multi-dimensional
structures, how are they revealed in the quantum harmonic oscillator case?


\subsection{Higher periodicity}
\label{sec:higherperiod}

The periodic reduction defined in section \ref{sec:periodicreduction} is part of
a more general family of periodic staircase initial value problems
\cite{papageorgiou1990integrable,capel1991complete,nijhoff1992lattice}.
In general, we define $2P$ initial conditions, $u_0, u_1, \ldots, u_{2P-1}$
such that $u_0 = \oh u_{2P-1}$,
along a staricase as shown in figure \ref{fig:staircase}.
The linearised KdV equation (\ref{eq:linear}) defines a dynamical map 
$(u_0,u_1,\ldots,u_{2P-1}) \to (\oh u_0,\oh u_1,\ldots,\oh u_{2P-1})$.
As before, we introduce reduced variables $x_1, \ldots, x_{P-1}$, 
$y_1, \ldots, y_{P-1}$ and can eliminate the $y_i$ to give 
a $P-1$ dimensional system of second order
difference equations in terms of the $x_i$ variables.

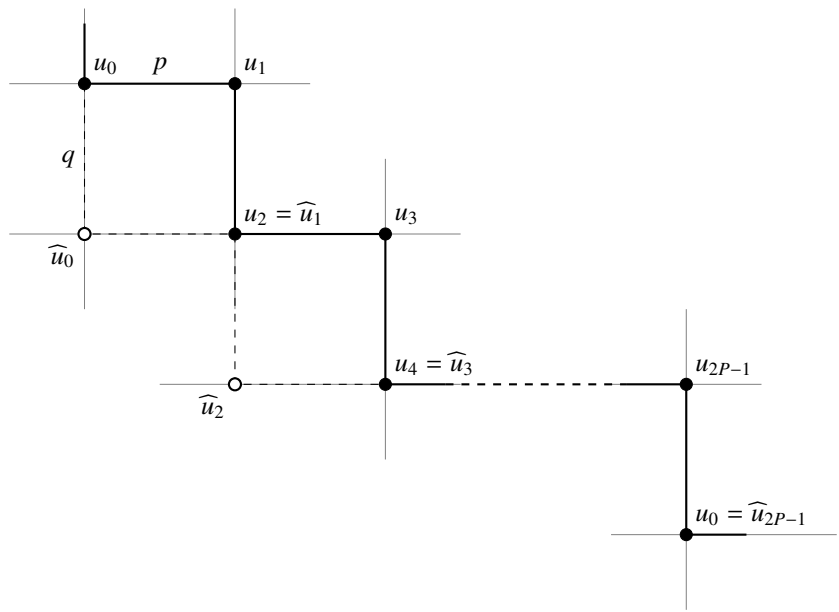
\begin{figure}[hbtp]
\begin{center}
    \begin{tikzpicture}  [scale=2]
      \begin{scope} [very thin,gray]
        \draw  (-0.5,0) -- (1.5,0) ;
        \draw  (-0.5,-1) -- (2.5,-1) ;
        \draw  (0.5,-2) -- (2.5,-2)  (3.5,-2) -- (4.5,-2) ;
        \draw (3.5,-3) -- (5,-3) ;
        \draw  (0,0.5) -- (0,-1.5) ;
        \draw  (1,0.5) -- (1,-1.5) ;
        \draw   (2, -0.5) -- (2,-2.5) ;
        \draw  (4,-1.5) -- (4,-3.5) ;
      \end{scope}
        \draw [thick] (0,0.4) -- (0,0) -- (1,0) -- (1,-1) -- (2,-1) -- (2,-2) -- (2.4,-2)   
                          (3.6,-2) -- (4,-2) -- (4,-3) -- (4.4,-3)  ;
        \draw [dashed] (0,0) -- (0,-1) -- (1,-1) -- (1,-2) -- (2,-2) ;
        \draw [thick, dashed]  (2.4,-2) -- (3.6,-2) ;      
        \filldraw [fill=black] circle (0.04) 
              node[anchor=south west] {$u_0$} ;
        \filldraw [fill=black] (1,0) circle (0.04)
              node[anchor = south west] {$u_1$} ;
        \filldraw [fill=black] (1,-1) circle (0.04) 
              node[anchor = south west] {$ u_2 = \oh u_1$} ;
        \filldraw [fill=black] (2,-1) circle (0.04) 
              node[anchor = south west] {$ u_3$} ; 
        \filldraw [fill=black] (2,-2) circle (0.04) 
              node[anchor = south west] {$ u_4 = \oh u_3$} ;                         
        \filldraw [fill=black] (4,-2) circle (0.04) 
              node[anchor = south west] {$ u_{2P-1}$} ;
        \filldraw [fill=black] (4,-3) circle (0.04) 
              node[anchor=south west] {$u_0 = \oh u_{2P-1}$} ;
        \filldraw [fill=white, thick] (0,-1) circle (0.04) 
              node[anchor=north east] {$\oh u_0$} ;    
        \filldraw [fill=white, thick] (1,-2) circle (0.04) 
              node[anchor=north east] {$\oh u_2$} ;   
        \draw (0.5,0) node [anchor=south] {$p$} ;
        \draw (0,-0.5) node [anchor=east] {$q$} ;   
    \end{tikzpicture}
\caption{The periodic staircase for period $P$.}
\label{fig:staircase}
\end{center}
\end{figure}

The $P=2$ case yields a 1 dimensional mapping that is entirely
equivalent to the case we have considered in section \ref{sec:periodicreduction},
except the lattice parameters combine in a slightly different way to give
the coefficient of the harmonic oscillator.

The $P=3$ case is the next case of interest, as here we find a system of
coupled harmonic oscillators in $x_1$ and $x_2$, with two commuting 
invariants and a similar commuting flow structure. In a similar manner
to (\ref{eq:p=1.5xb}) we can derive equations for a discrete flow in
variables $x_1$ and $x_2$:
\begin{equation}
    \oh x_1 + \oh x_2 + \uh {x_1} + s (2 x_1 + x_2) = 0
\ , \quad
    \oh x_2 + \uh {x_1} + \uh {x_2} + s (x_1 + 2 x_2) = 0
\ .
\label{eq:p3eqn1}
\end{equation}
As in section \ref{sec:commdiscflow}, we can also derive a commuting
flow for the evolution:
\begin{eqnarray}
    (1+tt') (\ob x_1 + \ub x_1) + \ob x_2 + tt' \ub x_2
    + (t+t') (2x_1 + x_2) &=& 0   \ ,
\label{eq:p3eqn2}
\\
    (1+tt') (\ob x_2 + \ub x_2) + tt' \ob x_1 + \ub x_1
    + (t+t') (x_1 + 2 x_2) &=& 0   \ .
\label{eq:p3eqn3}
\end{eqnarray}
Commutativity of these evolutions can be easily shown from
the first order form (with $x$ and $y$ variables) by writing each
evolution in matrix form; the resulting matrices commute.
The evolution then also possesses corner equations, which can be
derived using the eliminated $y$ variables. 
These allow us
to write closed form Lagrangians, such that $\Box \mcL = 0$ (\ref{eq:1formclosure})
on the equations of motion (\ref{eq:p3eqn1},\ref{eq:p3eqn2},\ref{eq:p3eqn3}):
\begin{eqnarray}
    \mcL_1 (x, \oh x) &=& 
    x_1 (\oh x_1 + \oh x_2) + x_2 \oh x_2
    + \tfrac12 s (x_1^2 + x_1 x_2 + x_2^2
    + \oh x_1^2 + \oh x_1 \oh x_2 + \oh x_2^2) \ ,
\label{eq:p3lag1}
\\
    \mcL_2 (x, \ob x) &=& 
    \frac{1 + tt'}{1 - tt'} (x_1 \ob x_1 + + x_2 \ob x_2)
    + \frac1{1-tt'} (x_1 \ob x_2 + tt' x_2 \ob x_1) 
\notag \\  && \quad 
    + \frac12 \frac{t + t'}{1 - tt'} (x_1^2 + x_1 x_2 + x_2^2
    + \ob x_1^2 + \ob x_1 \ob x_2 + \ob x_2^2) \ ,
\label{eq:p3lag2}
\end{eqnarray}
recalling the relation of $s,t,t'$.
A Lagrangian 1-form structure as in section \ref{sec:1form} follows.
Note that $\mcL_2$ represents a B\"acklund transform with parameter $r$.

The Lagrangians (\ref{eq:p3lag1},\ref{eq:p3lag2}) allow us to define
the momenta conjugate to $x_1, x_2$ writing 
$X_i = - \partial \mcL_1 / \partial x_i$,
\begin{equation}
    X_1 = -\big( \oh x_1 + \oh x_2 + \tfrac12 s (2x_1 + x_2) \big)
\ , \quad
    X_2 = -\big( \oh x_2 + \tfrac12 s (x_1 + 2 x_2) \big) \ .
\label{eq:p3momenta}
\end{equation}
with respect to which we have the invariant Poisson structure
$\{ x_i, X_j \} = \delta_{ij}$, preserved under the mappings. 
We could also write expressions for $X_i$ using $\mcL_2$, with
equality of these expressions producing the corner equations.

We can additionally derive two quadratic invariants of the mapping $I_1, I_2$,
which are invariant under both maps.
The canonical structure of (\ref{eq:p3momenta}) allows us to show
the critical integrability property that the two invariants
are in involution with each
other, with respect to the canonical Poisson bracket:
$
\{ I_1,I_2 \} =0
$ where
\begin{equation}
    I_1 = x_1 X_1 - 2 x_1 X_2 + 2 x_2 X_1 - x_2 X_2 
\ , \quad
    I_2 = \left( 1 - \tfrac34 s^2 \right)
    (x_1^2 + x_1 x_2 + x_2^2)
    + X_1^2 - X_1 X_2 + X_2^2  \ .
\end{equation}
The invariance and involutivity of these can be shown by direct calculation.
$I_1$ and $I_2$ will thus generate two commuting continuous flows
to the mapping.

For both the hat and the bar evolutions (\ref{eq:p3eqn1}),
(\ref{eq:p3eqn2}), (\ref{eq:p3eqn3}) it is possible to write
explicit solutions, and indeed we can find a joint solution
to the discrete evolutions:
\begin{equation}
    x_2 (m,n) = a \cos (\mu_+ m + \nu_+ n) + b \sin (\mu_+ m + \nu_+ n)
    + c \cos (\mu_- m + \nu_- n) + b \sin (\mu_- m + \nu_- n) \ ,
\label{eq:p3solution}
\end{equation}
where
$
  \cos \mu_{\pm} = -3s/4 \pm \tfrac12 \sqrt{1-3s^2/4}
$
and
\begin{equation}
  \cos \nu_{\pm} = -\frac{3(t+t')(1+tt')}{4(1+tt'+t^2t'^2)}
  \pm \frac12 \left( \frac{1-tt'}{1+tt'+t^2t'^2}\right )^2 
  \sqrt{1+tt'+t^2t'^2- \tfrac34 (t+t')^2} \ .
\end{equation}
We find $x_1 (m,n)$ similarly as a linear combination of shifts of $x_2$.
By considering derivatives with respect to the parameters $s$ and $t$
(recalling $t'$ is not independent of $s,t$), we can therefore derive
commuting continuous flows from the solution structure (\ref{eq:p3solution}).
We observe then again the interchange between continuous and discrete
parameters and variables, as in the lower periodic case.
We expect this will lead to a continuous Lagrangian 1-form structure,
but defer further investigation to a later paper.



\section{The Quantum Reduction}

In section \ref{sec:1form}, the discrete harmonic oscillator
model, arising as a special reduction from the linearised lattice KdV equation \eqref{eq:linear}, 
albeit a simple linear model nonetheless displays commuting discrete flows.
In the classical case, the Lagrangian 1-form structure captures these
commuting flows in a variational principle. A natural question is: what is
the quantum analogue for such a structure? 
Since the harmonic oscillator is well known and understood, it forms a good first toy
model for investigating Lagrangian form structures at the quantum level.

Integrable quantum mappings, arising from the  quantisation of mapping reductions from
lattice equations, were constructed and studied within the framework of canonical 
quantization and (non-ultralocal) R-matrix structures in \cite{nijhoff1993quantization,berry1979quantum,
quispel1992integrable,nijhoff1992integrable,field2005exact}. 
In a pioneering paper \cite{dirac1933lagrangian} Dirac took the position that the Lagrangian
approach to Physics is the more \emph{natural} one and proposed the first steps towards 
incorporating the Lagrangian into quantum mechanics, a route that was later pursued by Feynman 
leading to his concept of the path integral \cite{feynman1948space}. Concurring 
with Dirac's point of view, we seek here to understand the extended Lagrangian multiform  
variational principle on the quantum level, leading naturally to problem of finding a path integral 
version of that formalism in order to capture its natural quantum analogue. To make first steps in 
that direction the simple case of the quantum mappings derived in the previous section is a good 
starting point, exploiting the well-known formal techniques of path integrals, cf. e.g.  
\cite{feynman1965quantum,grosche1998handbook,schulman2005techniques}. As we will point out later
there are some similarities with ideas developed by Rovelli in \cite{rovelli2011discretizing,rovelli2011structure} 
who also uses the harmonic oscillator to develop ideas on reparametrisation invariant discretisations within 
the path integral framework, in particular the natural emergence of conservation of
the energy of the coninuous model within a time-slicing discretisation.

\subsection{Feynman Propagators}
\label{sec:unitary}

Beginning from our Lagrangian $\mcL_b$ (\ref{eq:p=1.5lagrangians})
we write the conjugate momenta $X:=X_b$ (\ref{eq:momenta}) and
$
       \oh X = \partial \mathcal L_b / \partial \oh x 
$.
In canonical quantisation, position $x$ and momentum $X$ become
\emph{operators} $\ox$ and $\oX$, such that
$ [\ox,\oX]=i\hbar $.
The momentum equations
(\ref{eq:momenta}) become operator equations of motion:
\begin{equation}
    \oh \ox - \ox =  
    \frac q{P-Q} \oh \oX - \frac{2P}{P-Q} \ox 
\ , \quad
    \oh \oX - \oX =
    - \frac{4Pq}{P-Q} \ox + \frac{2P}{P-Q} \oh \oX \ .
\label{eq:p=1.5opeqn}
\end{equation}
To understand the discrete time 
evolution we wish to express the evolution 
$(\ox,\oX) \to (\oh \ox, \oh \oX)$,
 in terms of a time-evolution operator $U_b$,
such that
$
    \ox \to \oh \ox = U_b^{-1} \ox U_b
$, $
    \oX \to \oh \oX = U_b^{-1} \oX U_b
$.
This is a canonical approach to discrete quantisation,
see for example \cite{nijhoff1993quantization}. 
Considering (\ref{eq:p=1.5opeqn}), it is not hard to see that an
appropriate choice of $U_b$ is given by:
\begin{equation}
    U_b = e^{i V(\ox)/2\hbar} e^{i T(\oX)/\hbar} e^{i V(\ox)/2\hbar} 
\ 
    = \exp \left(\frac{iP \ox^2}{\hbar q} \right)
       \exp \left(\frac{iq \oX^2}{2\hbar (P+Q)} \right)
       \exp \left(\frac{iP \ox^2}{\hbar q} \right) \ .
\label{eq:p=1.5U}
\end{equation}
In other words, a separated form for $U_b$ exists, but it is required
to have three terms. Note that (\ref{eq:p=1.5U}) is not a unique
form for $U_b$.

In discrete time, the one time-step propagator is then given by
$
    K_b (x, n; \oh x, n+1) = \,_{n+1} \langle \oh x | x \rangle_n
                   = \langle \oh x | U_b | x \rangle
$,
where we have moved in the second equality from time-dependent,
Heisenberg picture eigenstates to time-independent,
Schr\"odinger picture eigenstates. 
Since we have an
explicit form for $U_b$, we can calculate this expression by inserting
a complete set of momentum eigenstates:
\begin{eqnarray}
    \langle \oh x | U_b | x \rangle &=&
    \int \ud X
    e^{i V(\oh x)/2\hbar}
        \langle \oh x |  X \rangle
    e^{i T(X)/\hbar}
        \langle X | x \rangle
        e^{i V(x)/2\hbar}  \ ,
\notag \\
    &=& \left( \frac{i(P+Q)}{2 \pi \hbar q} \right)^{1/2}
        \exp \bigg\{ \frac i{\hbar q} \Big(
          (P+Q) x \oh x 
       + \tfrac12 (P-Q) (x^2 + \oh x^2)
                   \Big)  \bigg\} \ ,
\notag  \\
    &=& 
    \left( \frac{i(P+Q)}{2 \pi \hbar q} \right)^{1/2}
    \exp \left[ \frac i{\hbar} \mathcal L_b (x, \oh x) \right] \ .
\label{eq:onestepprop}
\end{eqnarray}
The second line results from a Gaussian integral:
the linearity of our system justifies taking the integration
region over the whole real line (we make some assumptions
here on the Hilbert space).
The final line recalls the Lagrangian (\ref{eq:p=1.5lagrangians}).
This is what might be expected for a ``one-step" path integral
(such as in \cite{field2005quantisation,field2006time}) noting that this approach
also specifies the normalisation constant.

This is sufficient to define the discrete-time path integral. By iterating
(\ref{eq:onestepprop}) over $N$ steps, we can write precisely the 
propagator for our discrete system:
\begin{equation}
    K_b (x_0, 0; x_N, N) =
    \left( \frac{i(P+Q)}{2 \pi \hbar q} \right)^{N/2}
    \prod_{n=1}^{N-1} \int_{-\infty}^{\infty} \ud x_n \
    e^{i \mathcal S[x(n)]/\hbar}  
\ , \quad
    \mathcal S [x(n)] = \sum_{n=0}^{N-1} \mathcal L_b (x_n , x_{n+1}) \ .
\label{eq:discretepathintegral}
\end{equation}
In this discrete case, equation (\ref{eq:discretepathintegral}) gives a 
precise definition to the path integral notation:
\begin{equation}
    K_b (x_0, 0; x_N, N) = \int_{x(0)=x_0}^{x(N)=x_N} [\mcD x(n)] 
                               \ e^{i \mathcal S[x(n)]/\hbar}  \ .
\end{equation}
Notice in particular that the normalisation associated to the measure is here
unambiguous.  In our quadratic regime, we can now calculate this explicitly.
Details are given in \ref{app:discretepropagator}, but we first expand our quantum
variables around the classical path, where the classical action can be evaluated as
$
    \mathcal S_{cl} = \sqrt P  
    [ 2 x_0 x_N - (x_0^2 + x_N^2) \cos \mu N ]
    / \sin \mu N \ .
$
Evaluating the discrete path integral as a series of $N$ Gaussian integrations,
and recalling the normalisation constant in (\ref{eq:discretepathintegral}),
we calculate the propagator:
\begin{equation}
    K_b (x_0, 0; x_N, N) =   
    \left( \frac{i \sqrt P}{\pi \hbar \sin (\mu N)} \right)^{1/2}
    \exp \left\{
    \frac{i \sqrt P}{\hbar \sin (\mu N)}
    \left(2x_0 x_N - (x_0^2 + x_N^2) \cos (\mu N) \right) \right\} \ .
\label{eq:discretepropagator}
\end{equation}
Note that this has the same form as the propagator for the continuous
time harmonic oscillator.
Dependence on the parameter $b$ is evident through 
$\cos \mu = -b$.
We note, then, that the propagator is common to both the discrete flow
and to the interpolating continuous time flow.

Using the operator equations of motion (\ref{eq:p=1.5opeqn}),
 it is easy to see that we have an operator invariant:
\begin{equation}
    \oI_b = \frac12 \oX^2 + 2P \ox^2 
\ 
    =
    \frac12 \left(- \hbar^2 \frac{\partial^2}{\partial x^2} + 4P x^2 \right) \ ,
\label{eq:qinvariant}
\end{equation}
This is, of course, simply the operator version of the classical invariant
(\ref{eq:p=1.5invariant}), and is precisely the Hamiltonian for the
continuous time harmonic oscillator, where $4P = \omega^2$.
Note that $\oI_b$ is $Q$ independent, and so it is clear that the 
same process applied to the bar evolution generated by $\mathcal L_a$ will give the same result.
In other words, both discrete quantum evolutions share the same invariant, which 
is the harmonic oscillator.
The invariant can also be considered from the perspective of path
integrals and the unitary operator following the method of \cite{field2006time};
this is elaborated in \ref{app:quantuminvariants}.
We can relate $\oI_b$ (\ref{eq:qinvariant}) to the evolution operator
$U_b$ (\ref{eq:p=1.5U}) in principle by a Campbell-Baker-Hausdorff
expansion (\cite{oteo1991baker,varadarajan2013lie});
an explicit form is given by algebraic manipulation:
\begin{equation}
    U_b = \exp \left[\frac1{\hbar \sqrt{P}}
    \arctanh \left( \frac{i \sqrt{P}}q
    \right) \ \oI_b \right] \ .
\end{equation}
So we can see clearly how the discrete quantum evolution
relates to a continuous time flow.




\subsection{Path independence of the propagator}
\label{sec:pathindepprop}

In equation (\ref{eq:discretepropagator}) we have established the 
propagator for an evolution in one discrete time variable; but
we have in the classical case two compatible discrete flows
(\ref{eq:p=1.5lagrangians}).
The one-step propagator in the hat direction is given in (\ref{eq:onestepprop}),
whilst in the bar direction it is easily deduced by the same method:
\begin{equation}
    K_a (x, \ob x; 1)  = 
    \left( \frac{i(P+R)}{2 \pi \hbar r} \right)^{1/2}
    \exp \left[ \frac i{\hbar} \mathcal L_a (x, \ob x) \right] \ .
\label{eq:onesteppropa}
\end{equation}
We remark that, as we have here a second time
direction, we might plausibly introduce a second $\hbar$ parameter. We ignore
such considerations for the time being and allow $\hbar$ to be the same in both
time directions.
In general, if we begin at a time co-ordinate $(0,0)$ and evolve
along integer time co-ordinates to a new
time $(N,M)$, the propagator could depend not only on the endpoints,
but also on the path $\Gamma$ taken through the time variables, see figure
\ref{fig:gamma}.
We associate to the path an action 
$\mcS_{\Gamma}:=\mcS [x(\bn);\Gamma]$ (\ref{eq:SGamma}). We can then define
a propagator for the evolution along the time-path $\Gamma$, made up of 
the one-step elements (\ref{eq:onestepprop}), (\ref{eq:onesteppropa}):
\begin{equation}
    K_{\Gamma} \big( x_b, (N,M); x_a,(0,0) \big)
    :=
    \mcN_{\Gamma} \prod_{(n,m) \in \Gamma} \int \ud x_{n,m} \
    \exp \left[ \frac i{\hbar} \mcS_{\Gamma} [x(\bn)] \right] \ ,
\label{eq:gammapropagator}
\end{equation}
where we have integrated over all internal points $x_{n,m}$ on the curve
$\Gamma$. Here $\mcN_{\Gamma}$ represents the product of normalisation factors from
the relevant elements of (\ref{eq:onestepprop}), (\ref{eq:onesteppropa}).

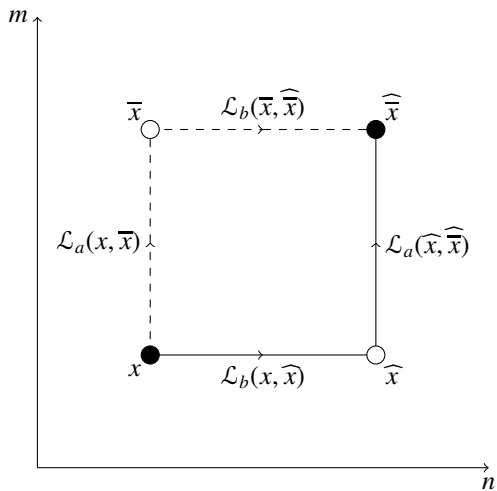
\begin{figure}[hbtp]
\begin{center}
\begin{tikzpicture}  [scale=3]

    \draw [->] (0,0) -- (2,0)   node [anchor=north] {$n$}   ;
    \draw [->] (0,0) -- (0,2)   node [anchor=east] {$m$} ;

    \draw [->] (0.5,0.5) -- (1,0.5) ;
    \draw [->] (1,0.5) -- (1.5,0.5) -- (1.5,1) ;
    \draw (1.5,1) -- (1.5,1.5) ;
    
    \draw [dashed,->] (0.5,0.5) -- (0.5,1) ;
    \draw [dashed,->] (0.5,1) -- (0.5,1.5) -- (1,1.5) ;
    \draw [dashed] (1,1.5) -- (1.5,1.5) ;

    \filldraw [fill=black] (0.5,0.5) circle (0.04) 
              node[anchor=north east] {$x$};
    \filldraw [fill=white] (1.5,0.5) circle (0.04) 
              node[anchor=north west] {$\oh x$} ;
    \filldraw [fill=black] (1.5,1.5) circle (0.04) 
              node[anchor=south west] {$\oh {\ob x}$} ;
    \filldraw [fill=white] (0.5,1.5) circle (0.04)
              node[anchor=south east] {$\ob x$} ;
    
    \draw (1,0.5) node[anchor=north]
           {$\mcL_b (x, \oh x)$} ;
    \draw (1.5,1) node[anchor=west]
           {$\mcL_a (\oh x, \oh {\ob x})$} ;
    \draw (0.5,1) node[anchor=east]
           {$\mcL_a (x, \ob x)$} ;
    \draw (1,1.5) node[anchor=south]
           {$\mcL_b (\ob x, \oh{\ob x})$} ;

\end{tikzpicture}
\caption{The solid line shows path (i) for $K_{\lrcorner}$,
and the dashed line path (ii) for $K_{\ulcorner}$. The white 
circles represent variables that are integrated over.}
\label{fig:Kcorners}
\end{center}
\end{figure}

We begin by considering the simple case of an evolution of one step in each direction.
There are two routes to achieve this, as shown in figure \ref{fig:Kcorners}. Either we
evolve first in the hat direction, followed by an evolution in the bar direction, or vice versa.
In path (i), we evolve first according to the hat evolution $\mathcal L_b$, and then according to
the bar evolution $\mathcal L_a$. We evaluate the propagator as:
\begin{equation}
    K_{\lrcorner} (x, \oh {\ob x}) = 
    \left( \frac{(P+Q)(P+R)}{(-2 \pi i \hbar)^2 qr} \right)^{1/2}
    \int^{\infty}_{-\infty} \ud \oh x \
    \exp \left \{ \frac i{\hbar} \left( \mathcal L_b (x, \oh x) 
                + \mathcal L_a (\oh x, \oh {\ob x}) \right) \right \} \ .
\label{eq:Klrcorner}
\end{equation}
For the alternative path (ii) we evolve first by the bar evolution $\mathcal L_a$,
and then the hat evolution $\mathcal L_b$:
\begin{equation}
    K_{\ulcorner} (x, \oh {\ob x}) = 
    \left( \frac{(P+Q)(P+R)}{(-2 \pi i \hbar)^2 qr} \right)^{1/2}
    \int^{\infty}_{-\infty} \ud \ob x \
    \exp \left \{ \frac i{\hbar} \left( \mathcal L_a (x, \ob x) 
                + \mathcal L_b (\ob x, \oh {\ob x}) \right) \right \} \ .
\label{eq:Kulcorner}
\end{equation}
These are both resolved by substituting Lagrangians (\ref{eq:p=1.5lagrangians})
and evaluating the Gaussian integral. 
The result is totally symmetric under interchange of the parameters 
$q$ and $r$, as are (\ref{eq:Klrcorner}) and (\ref{eq:Kulcorner});
so that
\begin{equation}
   K_{\ulcorner} (x, \oh{\ob x}) = K_{\lrcorner} (x, \oh{\ob x}) \ .
\label{eq:Kul=Kur}
\end{equation}
We find the same propagator for either path.
It is an obvious corollary of this result that, so long as we take only forward steps in 
time, the propagator $K_{N,M}(x_a, x_b)$ is independent of the path taken in the time
variables.

We could also consider a path in the time variables allowing backward 
time steps. As in the classical case, we can construct an action
for such a trajectory, using an appropriate orientation for the Lagrangians.
In the quantum
case we perform a path integral over this action, integrating over all 
intermediate points. As $U_b$ generates a time-step in the $b$ direction 
(section \ref{sec:unitary}), $U_b^{-1}$ generates the backward 
evolution.

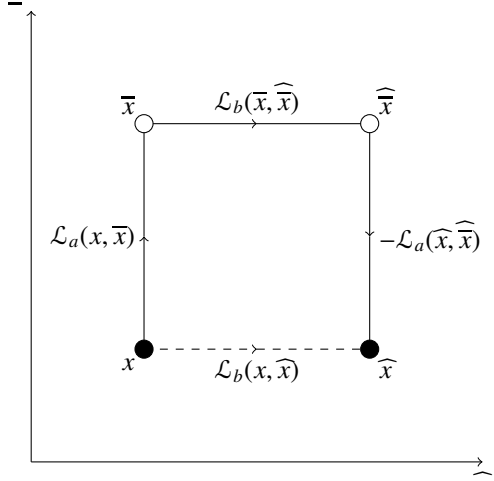
\begin{figure}[htp]
\begin{center}
\begin{tikzpicture}  [scale=3]
    \draw [->] (0,0) -- (2,0)   node [anchor=north] {$\oh{\phantom u}$}   ;
    \draw [->] (0,0) -- (0,2)   node [anchor=east] {$\ob{\phantom u}$} ;

  \begin{scope} [dashed]
    \draw [->] (0.5,0.5) -- (1,0.5) ;
    \draw (1,0.5) -- (1.5,0.5) ;
  \end{scope}
    
    \draw [->] (0.5,0.5) -- (0.5,1) ;
    \draw [->] (0.5,1) -- (0.5,1.5) -- (1,1.5) ;
    \draw [->] (1,1.5) -- (1.5,1.5) -- (1.5,1) ;
    \draw (1.5,1) -- (1.5,0.5) ;

    \filldraw [fill=black] (0.5,0.5) circle (0.04) 
              node[anchor=north east] {$x$};
    \filldraw [fill=black] (1.5,0.5) circle (0.04) 
              node[anchor=north west] {$\oh x$} ;
    \filldraw [fill=white] (1.5,1.5) circle (0.04) 
              node[anchor=south west] {$\oh {\ob x}$} ;
    \filldraw [fill=white] (0.5,1.5) circle (0.04)
              node[anchor=south east] {$\ob x$} ;
    
    \draw (1,0.5) node[anchor=north]
           {$\mcL_b (x, \oh x)$} ;
    \draw (1.5,1) node[anchor=west]
           {$- \mcL_a (\oh x, \oh {\ob x})$} ;
    \draw (0.5,1) node[anchor=east]
           {$\mcL_a (x, \ob x)$} ;
    \draw (1,1.5) node[anchor=south]
           {$\mcL_b (\ob x, \oh{\ob x})$} ;
\end{tikzpicture}
\caption{The path for action $\mcS_{\sqcap}$. In the propagator,
we integrate over the variables at the white circles.
Note the minus sign on the backward step, $-\mcL_a (\oh x, \oh {\ob x})$.}
\label{fig:backwardsstep}
\end{center}
\end{figure}

Considering once more the simplest case, we imagine a trajectory around three
sides of a square, shown in figure \ref{fig:backwardsstep}.
Including the normalisation factors from (\ref{eq:onestepprop})
this is described by the propagator,
\begin{equation}
    K_{\sqcap} (x, \oh x) = 
    \frac{(P+Q)^{1/2}(P+R)}{(2 \pi \hbar)^{3/2} (-iq)^{1/2} r}
    \int_{-\infty}^{\infty} \ud \ob x 
    \int_{-\infty}^{\infty} \ud \oh{\ob x} \ 
    \exp
    \left( \tfrac{i}{\hbar}
    \mathcal L_a (x, \ob x) 
    + \mathcal L_b (\ob x, \oh{\ob x}) - \mathcal L_a (\oh x, \oh{\ob x})
      \right) \ .
\end{equation}
This is easily calculated by Gaussian integrals, and yields:
\begin{equation}
    K_{\sqcap} (x, \oh x) = 
    \left( \frac{i(P+Q)}{2 \pi \hbar q} \right)^{1/2}
    \exp \left( \frac i{\hbar} \mathcal L_b (x, \oh x)  \right) 
    =
    K_b (x, \oh x; 1)    \ .
\label{eq:Ksq=Kb}
\end{equation}
So we regain exactly our one step propagator from (\ref{eq:onestepprop}). 
Remarkably, we again achieve Lagrangian closure, but now on the quantum level.
Recall that classically Lagrangian closure depended upon the equations of
motion: here we have left the equations of motion behind, and yet this
key result still holds.

We could also consider the possibility of a \emph{loop} in the
discrete variables, illustrated in figure \ref{fig:timeloop}(i). 
We imagine some unspecified incoming and outgoing actions 
$\mcS_{in} (x_a, x_1)$ and $\mcS_{out} (x_5, x_b)$,
a simple loop in discrete steps, and five integration variables
$x_1, \ldots, x_5$. Note that we assign two integration variables
to the same vertex, as it is visited twice by the path: the following
calculation will justify this choice as the correct one.

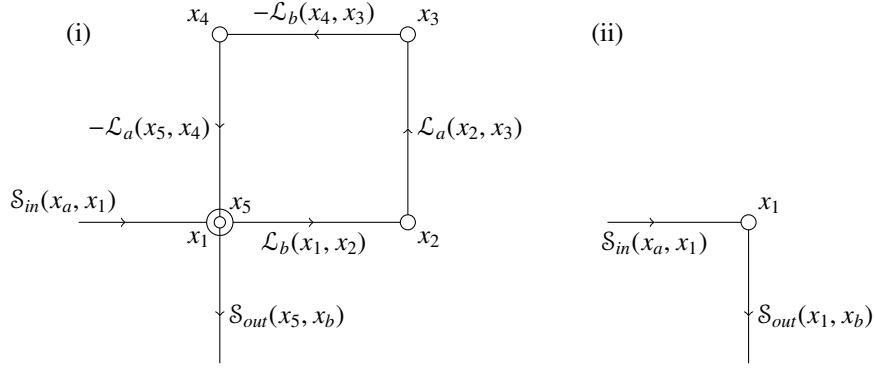
\begin{figure}[htbp]
\begin{center}
\begin{tikzpicture}  [scale=2.5]
    \draw [->] (-0.25,0.5) -- (0,0.5) ;
    \draw [->] (0,0.5) -- (1,0.5) ;
    \filldraw [fill=white] (0.5,0.5) circle (0.07)  
              node[anchor=north east] {$x_1$} ;
    \draw [->] (1,0.5) -- (1.5,0.5) -- (1.5,1);
    \draw [->] (1.5,1) -- (1.5,1.5) -- (1,1.5);
    \draw [->] (1,1.5) -- (0.5,1.5) -- (0.5,1) ;
    \draw [->] (0.5,1) -- (0.5,0) ;
    \draw (0.5,0) -- (0.5,-0.25) ;
    \filldraw [fill=white] (0.5,0.5) circle (0.03) 
              node[anchor=south west] {$x_5$};
    \filldraw [fill=white] (1.5,0.5) circle (0.04) 
              node[anchor=north west] {$x_2$} ;
    \filldraw [fill=white] (1.5,1.5) circle (0.04) 
              node[anchor=south west] {$x_3$} ;
    \filldraw [fill=white] (0.5,1.5) circle (0.04)
              node[anchor=south east] {$x_4$} ;
    \draw (1,0.5) node[anchor=north]
           {$\mcL_b (x_1, x_2)$} ;
    \draw (1.5,1) node[anchor=west]
           {$\mcL_a (x_2, x_3)$} ;
    \draw (0.5,1) node[anchor=east]
           {$-\mcL_a (x_5, x_4)$} ;
    \draw (1,1.5) node[anchor=south]
           {$-\mcL_b (x_4, x_3)$} ;
    \draw (0,0.5) node[anchor=south east]
           {$\mcS_{in} (x_a, x_1)$}  ;
    \draw (0.5,0) node[anchor=west]
           {$\mcS_{out} (x_5, x_b)$}  ;
    \draw (-0.25,1.5) node {(i)}  ;

\begin{scope}[xshift=80pt]
    \draw [->] (-0.25,0.5) -- (0,0.5) ;
    \draw [->] (0,0.5) -- (0.5,0.5) -- (0.5,0) ;
    \draw (0.5,0) -- (0.5,-0.25) ;
    \filldraw [fill=white] (0.5,0.5) circle (0.04) 
              node[anchor=south west] {$x_1$};
    \draw (0,0.5) node[anchor=north]
           {$\mcS_{in} (x_a, x_1)$}  ;
    \draw (0.5,0) node[anchor=west]
           {$\mcS_{out} (x_1, x_b)$}  ;
    \draw (-0.25,1.5) node {(ii)}  ;
\end{scope}
\end{tikzpicture}
\caption{(i) shows the loop in discrete variables. (ii) is what remains after
collapse of the loop.}
\label{fig:timeloop}
\end{center}
\end{figure}

We then consider the action for the loop,
$
    \mcS_{loop} = 
    \mcL_b (x_1,x_2) + \mcL_a (x_2,x_3)
    - \mcL_b (x_4,x_3) - \mcL_a (x_5,x_4)
$,
noting the orientations on the Lagrangians.
With normalising factors from (\ref{eq:onestepprop}) and complex
conjugations in the backward steps, we then have:
\begin{equation}
    K_{loop} (x_a, x_b) = \frac{P+Q}{2 \pi \hbar q}
    \frac{P+R}{2 \pi \hbar r} 
    \int \ud x_1 \ldots \int \ud x_5 \ 
    \exp \frac i{\hbar} \left\{
    \mcS_{in} + \mcS_{loop} + \mcS_{out}
    \right\} \ . 
\end{equation}
The $x_2$ and $x_4$ integrals are evaluated as in (\ref{eq:Klrcorner})
yielding,
\begin{multline}
    K_{loop} (x_a, x_b) = \frac{(P+Q)(P+R)}{2 \pi \hbar (P-qr)(q+r)}
    \iiint \ud x_1 \ud x_3 \ud x_5 \ 
    \exp \frac i{\hbar} \Bigg\{
    \mcS_{in} (x_a, x_1) + \mcS_{out} (x_5, x_b) 
\\
    - \frac{(P+Q)(P+R)}{(P-qr)(q+r)} (x_1 - x_5) x_3
    + \frac12 \left(
    \frac{P-qr}{q+r} - \frac{P(q+r)}{P-qr}
    \right) (x_1^2 - x_5^2)
    \Bigg\}  \ .
\end{multline}
The quadratic term in the exponent in $x_3$ disappears, and so the integral
$\ud x_3$ yields a Dirac delta function: $\delta (x_1 - x_5)$. Combined
with the integral over $x_5$ this forces $x_5=x_1$ (as expected) and
we finally conclude,
\begin{equation}
     K_{loop} (x_a, x_b) = 
     \int \ud x_1 \ 
     \exp \frac i{\hbar} \big\{
    \mcS_{in} (x_a, x_1) + \mcS_{out} (x_1, x_b)
    \big\} \ .
\label{eq:Kloop}
\end{equation}
Diagrammatically, this is equivalent to the disappearance of the loop,
shown in figure \ref{fig:timeloop}(ii).
Loops in the discrete variables therefore ``close'' and do not effect the 
overall propagator.

\begin{prop}
For the special choice of Lagrangians (\ref{eq:p=1.5lagrangians}),
the propagator along the time path $\Gamma$ (\ref{eq:gammapropagator})
is independent of the choice of $\Gamma$, depending only on the end points.
\end{prop}

\begin{pf}
Equations (\ref{eq:Kul=Kur}), (\ref{eq:Ksq=Kb}) and (\ref{eq:Kloop}) together
show that the propagator is unchanged under elementary deformations of the 
curve $\Gamma$. Since we have a simple topology, a curve $\Gamma_1$ can be 
deformed into any other curve $\Gamma_2$ 
(with the same endpoints) by a series of elementary deformations.
The proposition follows.
\end{pf}

The proposition now allows us to 
calculate the general propagator for $N$ steps in the hat direction and 
$M$ steps in the 
bar direction, compare (\ref{eq:gammapropagator}). 
We denote such a propagator from $x_a$ to $x_b$ 
by $K_{N,M} (x_a, x_b)$.
As a consequence of the path independence, it is then clear that
we can calculate this as
$
    K_{N,M} (x_a, x_b) = \int \ud x K_{N,0} (x_a, x) K_{0,M} (x, x_b)
$.
In other words, we can consider taking first all the hat-steps, followed by all
the bar-steps.
Taking our discrete propagator from (\ref{eq:discretepropagator}),
we can then carry out the integral as another Gaussian, but in fact the
result follows immediately from the group property of the propagator,
using its shared form with the continuous time case, so:
\begin{equation}
    K_{N,M} (x_a, x_b)    =
    \left( \frac{i \sqrt P}{\pi \hbar \sin (\mu N + \eta M)} \right)^{1/2}  
    \exp \left\{ \frac{i \sqrt P}{\hbar \sin (\mu N + \eta M)}
    \left( 2x_a x_b - (x_a^2 + x_b^2) \cos (\mu N + \eta M) \right) \right\} 
    \ , 
\label{eq:multipropagator}
\end{equation}
which bears a clear relation to the continuous time case.


\subsection{Uniqueness}
\label{sec:pathindepgeneral}

The time-path independence for the propagator of section 
\ref{sec:pathindepprop} is a special property of our choice of Lagrangian
(\ref{eq:p=1.5lagrangians}) that does not hold in general.
As classically the Lagrangian 1-form obeys the closure condition
(\ref{eq:1formclosure}), so in the
quantum case we have time-path independence of the propagators
as a natural quantum analogue. Whilst classically this closure holds
only on the equations of motion, in the quantum case the 
path-independence occurs as we perform the path integral over 
intermediate variables.
It emerges that, for given oscillator parameters $a$ and $b$,
there is a fairly unique choice of Lagrangians exhibiting time-path
independence.

Consider the generalised oscillator Lagrangians of equation 
(\ref{eq:generaloscillatorlagrangian}) and define propagators 
around two corners of a square, as in equations 
(\ref{eq:Klrcorner}) and (\ref{eq:Kulcorner}).
Here we allow $a$ and $b$ to be free oscillator parameters.
\begin{eqnarray}
    K_{\lrcorner} (x, \oh {\ob x}) &=& 
    \mcN_{\lrcorner} \int^{\infty}_{-\infty} \ud \oh x \
    \exp \left \{ \frac i{\hbar} \left( \mathcal L_b (x, \oh x) 
                + \mathcal L_a (\oh x, \oh {\ob x}) \right) \right \} \ ,
\label{eq:Klrcorner2}
\\
    K_{\ulcorner} (x, \oh {\ob x}) &=& 
    \mcN_{\ulcorner} \int^{\infty}_{-\infty} \ud \ob x \
    \exp \left \{ \frac i{\hbar} \left( \mathcal L_a (x, \ob x) 
                + \mathcal L_b (\ob x, \oh {\ob x}) \right) \right \} \ .
\label{eq:Kulcorner2}
\end{eqnarray}
$\mcN_{\lrcorner}$ and $\mcN_{\ulcorner}$ are undetermined
normalisation constants.
These paths are illustrated in figure \ref{fig:Kcorners}.

We demand equality of the exponents in these two expressions, once
the integral has been carried out; in other words we demand
$
    K_{\lrcorner} (x, \oh {\ob x}) =  K_{\ulcorner} (x, \oh {\ob x}) \ ,
$
up to a normalisation. Calculating these propagators via a Gaussian integral,
we then 
 derive conditions for time-path-independence on our coefficients, 
which can be found in \ref{app:pathindepgeneral}.
We find the necessary conditions on the coefficients:
\begin{equation}
    a_0 = \frac12 a + \frac f {2\alpha} 
\ , \quad
    b_0 = \frac12 b + \frac f {2\beta} 
\ , \quad
    \alpha = \frac{\gamma} {\sqrt{a^2-1}}
\ , \quad
    \beta = \frac{\gamma} {\sqrt{b^2-1}} \ .
\label{eq:pathindepconditions}
\end{equation}
As in (\ref{eq:closureconditions}) the constant $f$ makes no 
contribution and we ignore it. The general Lagrangians 
(\ref{eq:generaloscillatorlagrangian}) are therefore restricted to
a symmetric form, with a specified  overall constant given by the
oscillator parameters $a$, $b$.
Note that taking $a=(P-R)/(P+R)$, 
$b=(P-Q)/(P+Q)$ leads us to \emph{exactly} the conditions of (\ref{eq:closureconditions}) and the Lagrangians (\ref{eq:p=1.5lagrangians}). In conclusion:

\begin{prop}
For given oscillator parameters $a$ and $b$,
the Lagrangians (\ref{eq:p=1.5lagrangians}) are the unique Lagrangians,
up to constants $\gamma$ and $f$ (\ref{eq:closureconditions}),
such that the multi-time propagator is path independent.
\end{prop}

In other 
words, demanding time-path independence of the propagator is the natural
quantum analogue of the closure relation on the Lagrangian.


\subsection{Quantum Variational Principle: Lagrangian 1-form case}
\label{sec:qvariation1d}

Consider a quantum mechanical evolution from an initial time $(0,0)$ to a new
time $(N,M)$, along a time-path $\Gamma$: shown in figure \ref{fig:gamma}. 
We can consider a propagator for the evolution 
$K_{\Gamma} ( x_b; x_a)$ defined in (\ref{eq:gammapropagator}).
We have shown that,
in the special case of Lagrangians (\ref{eq:p=1.5lagrangians}),
the propagator defined above is independent of the path $\Gamma$
(it depends only on the endpoints); but that this is not true in general.
For a generic Lagrangian, $K_{\Gamma}$ will depend on the time-path chosen,
as shown in section \ref{sec:pathindepgeneral}.

Classically, the system is defined as the critical point for the variation of the action over not
only the dependent variables, but also over the independent variables, i.e., it is a critical point with respect to 
the variation of the time-path. This not only yields all the compatible equations of motion for the system, but also 
selects certain ``permissible'' Lagrangians which obey a closure relation (\ref{eq:1formclosure}). 
This then yields a system of extended EL equations of which the Lagrangian can be considered to be the solution, 
cf. \cite{yoo2011discreteCM}.

\begin{figure}[hbt]
\begin{center}
\begin{tikzpicture}  [scale=1.3]

    \draw [->] (0,0) -- (5,0)   node [anchor=north] {$n$}   ;
    \draw [->] (0,0) -- (0,5)   node [anchor=east] {$m$} ;

    \draw (1,1) -- (1,3) 
    node [anchor = east] {(a)}
    -- (4,3) ;

    \draw [dashed] (1,1) -- (1.5,1) -- (1.5,2) -- (2,2) -- (2,5)
        -- (2.5,5) -- (2.5,2.5) -- (4,2.5) -- (4,3)  ;
    \draw (2,4) node [anchor = east] {(b)}   ;

    \draw [thick,gray] (1,1) -- (1,0.5) -- (3.5,0.5) -- (3.5,1.5) -- (3,1.5) -- 
              (3,1) -- (4,1) 
               -- (4,2) -- (4.5,2) -- (4.5,3) -- (4,3)    ;
    \draw (4,1.5) node [anchor=west] {(c)}     ;

    \filldraw [fill=black] (1,1) circle (0.04) 
              node[anchor=north east] {$(0,0)$};
    \filldraw [fill=black] (4,3) circle (0.04) 
              node[anchor=south west] {$(N,M)$} ;    

\end{tikzpicture}
\caption{Three possible paths in the time-variables.
Path (a) is a direct path. 
Path (b) extends for some distance in the $m$ direction before returning.
Path (c) includes a loop in the time variables.}
\label{fig:timepaths}
\end{center}
\end{figure}
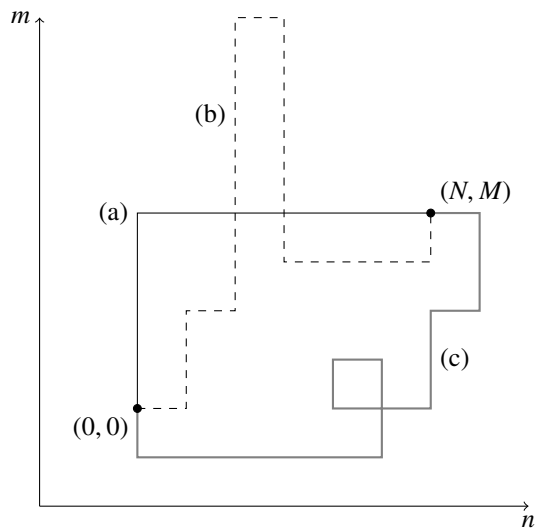

In the quantum case, we consider the dependence of the propagator on all possible
(discrete) time-paths $\Gamma$ between fixed initial and final times. 
In general, there are an infinite number of possible time paths from 
$(0,0)$ to $(N,M)$, including shortest time-paths as well as those with long
``diversions,'' or loops, as illustrated in figure \ref{fig:timepaths}. 
For a generic Lagrangian, as we vary the time path, each $\Gamma$  yields a different propagator
(\ref{eq:gammapropagator}) viewed as a functional of the path. 
In the special case of the Lagrangian (\ref{eq:p=1.5lagrangians}), 
however, the propagator $K_{\Gamma}$ is
\emph{independent of the path taken through the time variables},
and so remains unchanged across the variation of the time-path $\Gamma$.
This suggest that this path independence property is the natural
quantum analogue of the Lagrangian closure condition (\ref{eq:1formclosure}).

Pushing this idea one step further: 
viewing the propagator as a functional of the Lagrange function,
the Lagrangian itself can be thought of as representing a critical 
point (in a properly chosen function space of Lagrange functions) for the path-dependent propagator, 
with regard to variations of the time-path.
We suppose we can vary the path in 
such a way that the critical point analysis \emph{selects} the path independent Lagrangian 
from the space of possible Lagrangians (this was the point of view put forward in \cite{lobb2013variational} 
in the classical case). In a quantum setting this principle would be represented by a ``sum over all time-paths'' 
scenario, i.e. by means of posing a new quantum object of the form as was proposed in the continuous time-case 
in \cite{nijhoff2013lagrangian}. As a functional of the Lagrangian such an object would have a singular point 
for those Lagrangians which possess the quantum closure condition, i.e., those where the contributions of the 
path-independent propagators over which one integrates all contribute the same amount. How to control the 
singular behaviour of such an object is a matter of ongoing investigation.


\section{Quantisation of the Lattice Equation}

In section \ref{sec:lattice} we introduced the linear lattice equation
(\ref{eq:generallinear}). Having considered the quantisation of its finite
dimensional reduction, we now turn to quantisation of the lattice equation
itself. Quantisation of lattice models has been previously considered from
a canonical (quantum inverse scattering method) perspective 
\cite{volkov1992quantum,volkov1997quantum}, but here we will
bring a Lagrangian, path integral perspective to bear on this system.

Classically, we suppose the equation (\ref{eq:generallinear}) to hold on all plaquettes in the multidimensional
lattice at the same time. The equation is generated by the oriented Lagrangian:
\begin{equation}
    \mcL_{ij} (u, u_i, u_j; p_i, p_j) = 
      u (u_i - u_j) - \frac12 s_{ij} (u_i - u_j)^2 \ ,
    \qquad
      s_{ij} = \frac{p_i + p_j}{p_i - p_j} \ .
\label{eq:qlatticelagrange}
\end{equation}
The Lagrangian itself is a critical point of the classical variational principle over
surfaces: it obeys the closure property on the classical equations of motion, 
such that the surface can be allowed to freely vary under local moves. Indeed, it is 
also fairly unique, as seen in (\ref{eq:linearlagrangian2}).

How might we proceed to quantise such a system? A canonical approach is to transform
(\ref{eq:generallinear}) into an operator equation of motion, but we are concerned
here with a Lagrangian approach.
The clear analogy is to quantum field 
theory: we have a discretised space-time and a Lagrangian in two dimensions over field 
variables $u (\bn)$ indexed by a discrete vector $\bn$.
We imagine some space-time boundary $\partial \sigma$ enclosing a 
multi-dimensional
surface $\sigma$ made up of elementary plaquettes $\sigma_{ij}$. We can then construct an action
by summing the directed Lagrangians over the surface, as we would classically:
\begin{equation}
    \mcS_{\sigma} = \sum_{\sigma_{ij} \in \sigma} \mcL_{ij} (u,u_i, u_j) \ ,
\end{equation}
where we define the shorthand $\mcL_{ij} (u):= \mcL (u, u_i, u_j; p_i, p_j)$.

We then consider the propagator $K_{\sigma} (\partial \sigma)$,
 where all interior field
variables on the surface are integrated over. The propagator depends,
in principle, on the 
surface $\sigma$ and is a function of the field variables on the 
boundary $\partial \sigma$, which form some boundary value
problem
(see a similar point made in \cite{rovelli2011structure}):
\begin{equation}
    K_{\sigma} (\partial \sigma) = 
    \int [\mcD u (\bn)]_{\sigma} \ e^{i \mcS_{\sigma} [u (\bn)]/ \hbar}  
\ 
    =
    \mcN_{\sigma} \prod_{\bn \in \sigma} \int \ud u (\bn) \ e^{i \mcS_{\sigma} [u (\bn)]/ \hbar} \ .
\label{eq:correlation}
\end{equation}
We will see as we go on that this object is subject to infra-red divergences,
as particular surface configurations produce integrations yielding volume factors.
Since our main statements involves only the combinatorics of the exponential factors involving 
the action arising through Gaussian integrals, we tacitly assume $K_{\sigma}$ can be renormalised by an appropriate 
choice of  normalisation factor $\mcN_{\sigma}$. 
$K_{\sigma} (\partial \sigma)$ describes a propagator in the sense
of a surface gluing procedure: two propagators $K_{\sigma_1}$ and $K_{\sigma_2}$
are combined to a new propagator by multiplication and integration
over all variables living on the shared boundary 
$\partial \sigma_1 \cap \partial \sigma_2$. 
Thus, the one-step surface gluing can be written symbolically as 
\begin{equation}
  K_{\sigma_1 \cup \sigma_2} = 
  \int_{\partial \sigma_1 \cap \partial \sigma_2} K_{\sigma_1} \ast K_{\sigma_2} \ 
  :=
  \mcN_{\partial \sigma_1 \cap \partial \sigma_2}
  \left[ \prod_{\bn \in \partial \sigma_1 \cap \partial \sigma_2}
  \int \ud u(\bn)  \ \right] \
  K_{\sigma_1} (\partial \sigma_1) . K_{\sigma_2} (\partial \sigma_2)
  \ ,
\end{equation} 
where the integral is over appropriately chosen coordinates of the joined boundary. Iterating the gluing formula 
is tantamount to setting up a ``surface-slicing'' procedure for the path integral.

\subsection{Motivation: The pop-up cube}

Classically, for a Lagrangian 2-form we vary the surface $\sigma$ so that the 
Lagrangian and equations of motion sit at a critical point: the action should be
invariant under the variation of not only the dependent variables $u$, but also the
variation of the surface itself. As we move to the quantum regime, we then naturally
ask what happens to our propagator $K_{\sigma} (\partial \sigma)$ 
(\ref{eq:correlation}) under variation 
of the surface $\sigma$? We consider the effect of a simple variation of the
surface: from a flat surface to a popped-up cube, see figure \ref{fig:popup}.

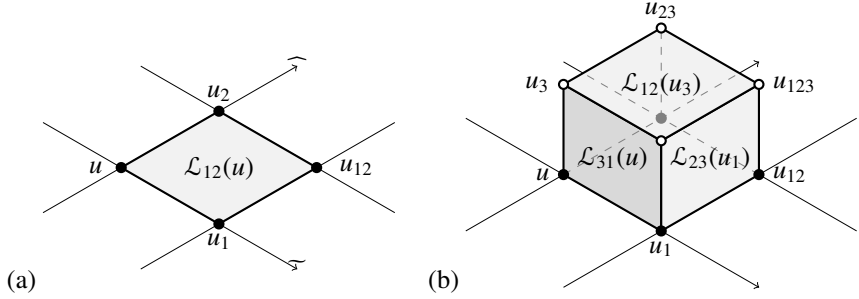
\begin{figure}[htp]
\begin{center}
(a)
\begin{tikzpicture}  [scale=1.5]
  \begin{scope} [very thin]
    \draw [->]  (30:-0.8) -- (30:1.8) node [black] {$\widehat{\phantom u}$} ;
    \draw [->]  (150:0.8) -- (150:-1.8) node [black] {$\widetilde{\phantom u}$} ;
    \draw   (330:1)  +(30:-0.8) -- +(30:1.8) ;
    \draw   (30:1)  +(150:0.8) -- +(150:-1.8) ;
  \end{scope}
    \filldraw [fill=black!5!white , draw =black, thick] 
              (0,0) -- ++ (30:1) -- ++(330:1) -- ++(210:1)  -- cycle  ;
    \filldraw [fill=black, thick] circle (0.04) 
          node[anchor=east] {$u \, \,$} ;
    \filldraw [fill=black, thick] (30:1) circle (0.04) 
          node[anchor=south] {$u_2$} ;    
    \filldraw [fill=black, thick] (330:1) circle (0.04)
          node[anchor = north] {$u_1$} ;
    \filldraw [fill=black, thick] (30:1) ++(330:1) circle (0.04) 
          node[anchor = west] {$ \, \, \, u_{12}$} ;
    \draw (30:1) +(0,-0.5) node {$\mathcal L_{12}(u)$}  ;
\end{tikzpicture}
\quad
(b)
\begin{tikzpicture}  [scale=1.5]
  \begin{scope} [very thin]
    \draw [->]  (30:-1) -- (30:2)  ;
    \draw [->]  (150:1) -- (150:-2)  ;
    \draw   (330:1)  +(30:-1) -- +(30:2) ;
    \draw   (30:1)  +(150:1) -- +(150:-2) ;
  \end{scope}
    \filldraw [fill=black!15!white , thick]  (0,0) -- ++(0,0.8) -- ++ (330:1) -- ++(0,-0.8) -- cycle  ; 
    \filldraw [fill=black!5!white , thick]  (330:1) -- ++(0,0.8) -- ++(30:1) -- ++(0,-0.8) -- cycle  ;
      
    \filldraw  [fill=black!5!white , thick] [yshift=0.8cm]
              (0,0) -- ++ (30:1) -- ++(330:1) -- ++(210:1)  -- cycle  ;
    \draw [dashed, very thin, gray] (0,0) -- (30:1.8) ;
    \draw [dashed, very thin, gray] (30:1) +(330:1) -- +(150:0.8) ;
    \draw [dashed, very thin, gray] (30:1) -- +(0,0.8) ;
    \filldraw [fill=black, thick] circle (0.04) 
          node[anchor=east] {$u  \,$} ;
    \filldraw [fill=gray , draw=gray, thick] (30:1) circle (0.04) ;    
    \filldraw [fill=black, thick] (330:1) circle (0.04)
          node[anchor = north] {$u_1$} ;
    \filldraw [fill=black, thick] (30:1) ++(330:1) circle (0.04) 
          node[anchor = west] {$ \, u_{12}$} ;
  \begin{scope}  [yshift=0.8cm]  
    \filldraw [fill=white, thick] circle (0.04)
          node [anchor=east] {$ u_3 \, $} ;
    \filldraw [fill=white, thick] (30:1) circle (0.04) 
          node [anchor=south] {$ u_{23}$} ;    
    \filldraw [fill=white, thick] (330:1) circle (0.04)
                      ;
    \filldraw [fill=white, thick] (30:1) ++(330:1) circle (0.04)
          node [anchor=west] {$ \, u_{123} $} ;
  \end{scope}
    \draw (30:1) +(0,0.3) node {$\mathcal L_{12} (u_3) $}  ;
    \draw (330:0.5) +(0,0.4) node {$\mathcal L_{31} (u)$}  ;
    \draw (330:1) ++(30:0.5) +(0,0.4) node {$\mathcal L_{23} (u_1)$}  ;
\end{tikzpicture}
\end{center}
\caption{A flat surface in (a), compared to a pop-up cube shown
in (b)}
\label{fig:popup}
\end{figure}

The contribution to the action given by surface (a) is a single Lagrangian,
$\mcL_{12} (u)$. In surface (b) we have five plaquettes, with a 
contribution to the action given by a sum of oriented Lagrangians:
$
    \mcS_{pop} [u_{n,m}] = \mcL_{23} (u_1) + \mcL_{31} (u_2) + \mcL_{12} (u_3)
      - \mcL_{23} (u) - \mcL_{31} (u)
$.
Note that the orientations lead to the negative contributions. 
In our path integral perspective (\ref{eq:correlation}), in (b) we must
also integrate over the ``popped-up'' variables $u_3$, $u_{23}$, 
$u_{31}$, $u_{123}$. 
The boundary variables on which the contributions depend are $u$,
$u_1$, $u_2$ and $u_{12}$.
So altogether, the contribution to the propagator for the pop-up cube is
like this:
\begin{equation}
    K_{pop} = \iiiint \ud u_{3} \ud u_{31} \ud u_{23} \ud u_{123} \
       \exp \left( \frac i{\hbar} \mcS_{pop} [u_{n,m}] \right) \ .
\label{eq:popup1}
\end{equation}
Now note that $\mcS_{pop} [u_{n,m}]$ contains no factor of $u_{123}$,
so that the integral $\int \ud u_{123}$ produces a volume factor $V$.
Equation (\ref{eq:popup1}) can then be written in a matricial form:
\begin{equation}
    K_{pop} = V \int \ud^3 \bu \
     \exp \frac i{\hbar} \left(\frac12 \bu^T A \bu + \bB^t \bu \ 
    + \frac12 
         \left[s_{31} (u_1^2 - u_{12}^2) + s_{23} (u_2^2 - u_{12}^2)
          + (u + u_{12})(u_1 - u_2)
       \right]        \right) \ ,
\label{eq:popup2}
\end{equation}
where
$
    \bu^T = (u_{3}, u_{31}, u_{23}) \ 
$, $ \ 
    \bB^T = ( -s_{31} u_1 - s_{23} u_2 , -u_1 + s_{23} u_{12} , 
                 u_2 + s_{31} u_{12} ) \ 
$, and
\begin{equation}
    A = 
        \left( \begin{array}{ccc}
                s_{23} + s_{31}  &  1  &  -1   \\
                1  &  -(s_{12} + s_{23})  &  s_{12}  \\
                -1  &  s_{12}  &  -(s_{12} + s_{31})
               \end{array}   \right)  \ .
\end{equation}
Now, in principle, equation (\ref{eq:popup2}) could be solved as a set of three
Gaussian integrals, but matrix $A$ is in fact singular. 
The parameter identity for $s_{ij}$ (\ref{eq:qlatticelagrange}):
\begin{equation}
    s_{12} s_{23} + s_{23} s_{31} + s_{31} s_{12} + 1 = 0 \ ,
\label{eq:sijidentity}
\end{equation}
leads to $\det A = 0$.   
We therefore resolve (\ref{eq:popup2}) by carrying out 
\emph{two} Gaussian integrals, knowing for the third integration variable we 
shall be left with an exponent that is at most linear.
Performing Gaussian integrations with respect to $u_3$ and $u_{31}$, we therefore have:
\begin{equation}
    K_{pop} = V \frac{2 \pi \hbar}{s_{23}} \int \ud u_{23}
              \exp \frac i{\hbar} \left(
              u (u_1 - u_2) - \tfrac12 s_{12} (u_1 - u_2)^2
              \right)
\ 
    =  V^2 \frac{2 \pi \hbar}{s_{23}}
           \exp \left( \frac i{\hbar} \mcL_{12} (u, u_1, u_2) \right) \ ,
\label{eq:popup3}
\end{equation}
where in the first equality we note that all terms containing $u_{23}$ have
vanished entirely.
This is now \emph{exactly} the exponent expected from the diagram (a) in figure
\ref{fig:popup}. So, whilst it is clear that there are non-trivial issues to
resolve with respect to volume factors and normalisation factors in (\ref{eq:popup3}),
\footnote{The asymmetrical factor of $s_{23}$ in the prefactor is an indicator that renormalisation requires some careful thought.}
in the critical issue of the contribution to the \emph{action} in the exponent
between diagrams \ref{fig:popup}(a) and \ref{fig:popup}(b), the two
pictures make the \emph{same} contibution. In other words, there is some sense
in which the action is unchanged by the local move that transforms the surface
$\sigma$ by the pop-up cube. Inspired by this discovery, we consider a more
general situation.

\subsection{Surface Independence of the propagator}
\label{sec:elementaryspecial}

In the classical case,  there are three
\emph{elementary configurations} of Lagrangians in three dimensions, that form the basis of all other possible configurations \cite{lobb2013variational}. 
We can attach to these  configurations
three \emph{elementary moves} in the quantum mechanical case that form the basis for
deformations of the surface $\sigma$. These elementary moves are shown in figures
\ref{fig:elementarya}, \ref{fig:elementaryb} and \ref{fig:elementaryc}.
Combined with the pop-up cube of figure \ref{fig:popup} these give a full set
of local moves for deforming the surface $\sigma$.

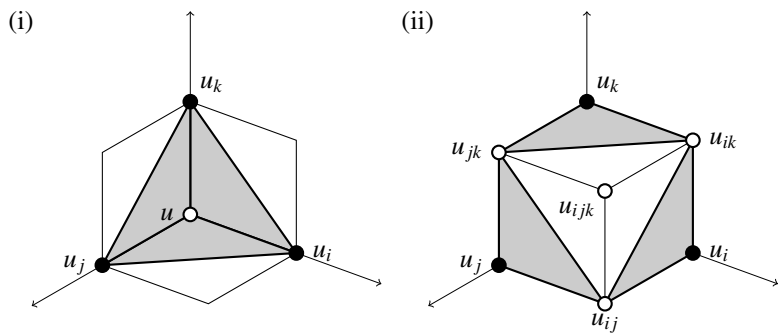
\begin{figure}[htp]
\begin{center}
\begin{tikzpicture}  [scale=1.5]
  \begin{scope} [very thin]
    \draw [->]  (0,0) -- (-20:1.8)  ;
    \draw [->]  (0,0) -- (210:1.62)  ;
    \draw [->]  (0,0) -- (0,1.8)    ;
  \end{scope}
    \draw (0,1) -- ++(340:1) -- ++(0,-1) -- 
          ++(210:0.9) -- ++ (160:1) -- ++ (0,1)
          -- cycle ;
    \filldraw [fill=black!20!white , draw =black, thick] 
         (0,1) -- (0,0) -- (340:1) -- cycle  ;
    \filldraw [fill=black!20!white , draw =black, thick] 
         (210:0.9) -- (0,0) -- (340:1) -- cycle  ;
    \filldraw [fill=black!20!white , draw =black, thick] 
         (0,1) -- (0,0) -- (210:0.9) -- cycle  ;
    \filldraw [fill=white, thick] circle (0.06) 
          node[anchor=east] {$u \ $} ;
    \filldraw [fill=black, thick] (210:0.9) circle (0.06)
          node[anchor = east] {$u_j \ $} ;
    \filldraw [fill=black, thick] (340:1) circle (0.06) 
          node[anchor=west] {$ \ u_i$} ;    
    \filldraw [fill=black, thick] (0,1)  circle (0.06) 
          node[anchor = south west] {$  u_k$} ;
   \draw (-1.5,1.7)  node {(i)} ;
   \begin{scope} [xshift=100pt]
  \begin{scope} [very thin]
    \draw [->]  (-20:1) -- (-20:1.8)  ;
    \draw [->]  (210:0.9) -- (210:1.62)  ;
    \draw [->]  (0,1) -- (0,1.8)    ;
  \end{scope}
    \draw (0,1) ++(340:1) ++(210:0.9) -- +(0,-1) ;
    \draw (0,1) ++(340:1) ++(210:0.9) -- +(160:1) ;
    \draw (0,1) ++(340:1) ++(210:0.9) -- +(30:0.9) ;
    \filldraw [fill=black!20!white , draw =black, thick] 
         (0,1) ++(210:0.9) -- (0,1) -- ++(-20:1) -- cycle ;
    \filldraw [fill=black!20!white , draw =black, thick] 
         (210:0.9) ++(0,1) -- (210:0.9) -- ++(340:1) -- cycle  ;
    \filldraw [fill=black!20!white , draw =black, thick] 
         (-20:1) ++(0,1) -- (-20:1) -- ++(210:0.9) -- cycle  ;
    \filldraw [fill=white, thick] (0,1) ++(340:1) ++(210:0.9)
          circle (0.06) 
          node[anchor=north east] {$u_{ijk} $} ;
    \filldraw [fill=white, thick] (0,1) ++(210:0.9) circle (0.06)
          node[anchor = east] {$u_{jk} \ $} ;
    \filldraw [fill=white, thick] (0,1) ++(340:1) circle (0.06) 
          node[anchor=west] {$ \ u_{ik}$} ;    
    \filldraw [fill=white, thick] (210:0.9)++(340:1)  circle (0.06) 
          node[anchor = north] {$ u_{ij}$} ;
    \filldraw [fill=black, thick] (210:0.9) circle (0.06)
          node[anchor = east] {$u_j \ $} ;
    \filldraw [fill=black, thick] (340:1) circle (0.06) 
          node[anchor=west] {$ \ u_i$} ;    
    \filldraw [fill=black, thick] (0,1)  circle (0.06) 
          node[anchor = south west] {$  u_k$} ;
   \draw (-1.5,1.7)  node {(ii)} ;
    \end{scope}
\end{tikzpicture}
\end{center}
\caption{Elementary move (a). We pass between (i) and (ii); white circles
indicate variables to be integrated over in the move.}
\label{fig:elementarya}
\end{figure}

The first move is shown in figure \ref{fig:elementarya}.
The action and contribution to the propagator (\ref{eq:correlation}) 
for figure \ref{fig:elementarya}(i) are given by:
\begin{equation}
    \mcS_{(ai)} = \mcL_{ij}(u) + \mcL_{jk} (u) + \mcL_{ki} (u) 
\ , \quad
    K_{(ai)} = \mcN_{(ai)} 
    \int \ud u \exp \left[ i \mcS_{(ai)} / \hbar \right] \ .
\label{eq:Ka1}
\end{equation}                   
In contrast, for figure \ref{fig:elementarya} (ii):
\begin{equation}
    \mcS_{(aii)} = \mcL_{ij}(u_k) + \mcL_{jk} (u_i) + \mcL_{ki} (u_j) 
\ , \quad
    K_{(aii)} = \mcN_{(aii)}
    \iiiint \ud u_{ij} \ud u_{jk} \ud u_{ki} \ud u_{ijk}
     \exp \left[ i \mcS_{(aii)} / \hbar \right] \ .
\label{eq:Ka2}
\end{equation}
We have some issue in both of these cases with volume factors appearing
in the evaluation; but we proceed under the assumption that these
can be dealt with through some regularisation and normalisation.
As shown in \ref{app:elmovea}, we then find that the \emph{exponents}
 in $K_{(ai)}$ and $K_{(aii)}$ are the same.
With the correct choice of normalisation and regularisation, we have
identical contributions to the propagator.

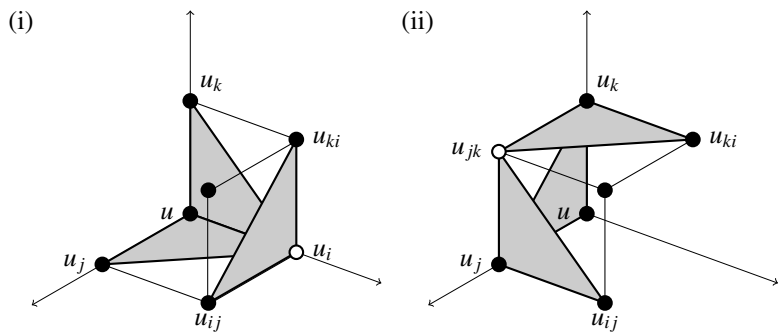
\begin{figure}[htp]
\begin{center}
   \begin{tikzpicture}  [scale=1.5]
  \begin{scope} [very thin]
    \draw [->]  (0,0) -- (-20:1.8)  ;
    \draw [->]  (0,0) -- (210:1.62)  ;
    \draw [->]  (0,0) -- (0,1.8)    ;
  \end{scope}
    \filldraw [fill=black!20!white , draw =black, thick] 
         (0,1) -- (0,0) -- (340:1) -- cycle  ;
    \filldraw [fill=black!20!white , draw =black, thick] 
         (210:0.9) -- (0,0) -- (340:1) -- cycle  ;
    \draw (-21:1) -- ++(0,1) -- ++(210:0.9) -- ++(0,-1) -- cycle;
    \draw (0,1) -- ++(-20:1) ;
    \draw (210:0.9) -- ++(-20:1) ;
    \filldraw [fill=black!20!white , draw =black, thick] 
         (-20:1) ++(0,1) -- (-20:1) -- ++(210:0.9) -- cycle  ;
    \filldraw [fill=black, thick] circle (0.06) 
          node[anchor=east] {$u \ $} ;
    \filldraw [fill=black, thick] (210:0.9) circle (0.06)
          node[anchor = east] {$u_j \ $} ;
    \filldraw [fill=white, thick] (340:1) circle (0.06) 
          node[anchor=west] {$ \ u_i$} ;    
    \filldraw [fill=black, thick] (0,1)  circle (0.06) 
          node[anchor = south west] {$  u_k$} ;
    \filldraw [fill=black, thick] (0,1) ++(340:1) circle (0.06) 
          node[anchor=west] {$ \ u_{ki}$} ;    
    \filldraw [fill=black, thick] (210:0.9)++(340:1)  circle (0.06) 
          node[anchor = north] {$ u_{ij}$} ;
    \filldraw [fill=black, thick] (0,1) ++(340:1) ++(210:0.9)
          circle (0.06) 
           ;
   \draw (-1.5,1.7)  node {(i)} ;
   \begin{scope} [xshift=100pt]
  \begin{scope} [very thin]
    \draw [->]  (0,0) -- (-20:1.8)  ;
    \draw [->]  (0,0) -- (210:1.62)  ;
    \draw [->]  (0,0) -- (0,1.8)    ;
  \end{scope}
    \filldraw [fill=black!20!white , draw =black, thick] 
         (0,1) -- (0,0) -- (210:0.9) -- cycle  ;
    \draw (0,1) ++(340:1) ++(210:0.9) -- +(0,-1) ;
    \draw (0,1) ++(340:1) ++(210:0.9) -- +(160:1) ;
    \draw (0,1) ++(340:1) ++(210:0.9) -- +(30:0.9) ;
    \filldraw [fill=black!20!white , draw =black, thick] 
         (0,1) ++(210:0.9) -- (0,1) -- ++(-20:1) -- cycle ;
    \filldraw [fill=black!20!white , draw =black, thick] 
         (210:0.9) ++(0,1) -- (210:0.9) -- ++(340:1) -- cycle  ;
    \filldraw [fill=black, thick] circle (0.06) 
          node[anchor=east] {$u \ $} ;
    \filldraw [fill=black, thick] (210:0.9) circle (0.06)
          node[anchor = east] {$u_j \ $} ;
    \filldraw [fill=black, thick] (0,1)  circle (0.06) 
          node[anchor = south west] {$  u_k$} ;
    \filldraw [fill=black, thick] (0,1) ++(340:1) circle (0.06) 
          node[anchor=west] {$ \ u_{ki}$} ;    
    \filldraw [fill=black, thick] (210:0.9)++(340:1)  circle (0.06) 
          node[anchor = north] {$ u_{ij}$} ;
    \filldraw [fill=white, thick] (0,1) ++(210:0.9) circle (0.06)
          node[anchor = east] {$u_{jk} \ $} ;
    \filldraw [fill=black, thick] (0,1) ++(340:1) ++(210:0.9)
          circle (0.06) 
           ;
   \draw (-1.5,1.7)  node {(ii)} ;
    \end{scope}   
\end{tikzpicture}
\end{center}
\caption{Elementary move (b). White circles indicate integration variables.}
\label{fig:elementaryb}
\end{figure}

We then consider elementary move (b), shown in figure
\ref{fig:elementaryb}. We have the action and propagator contribution 
for figure \ref{fig:elementaryb}(i):
\begin{equation}
    \mcS_{(bi)} = \mcL_{ij}(u) + \mcL_{ki} (u) - \mcL_{jk} (u_i) 
\ , \quad
    K_{(bi)} = \mcN_{(bi)}  \int \ud u_{i}
     \exp \left[ i \mcS_{(bi)} / \hbar \right] \ .
\end{equation}
Similarly for figure \ref{fig:elementaryb}(ii):
\begin{equation}
    \mcS_{(bii)} = \mcL_{ij}(u_k) + \mcL_{ki} (u_j) - \mcL_{jk} (u) 
\ , \quad
    K_{(bii)} = \mcN_{(bii)}  \int \ud u_{jk}
     \exp \left[ i \mcS_{(bii)} / \hbar \right] \ .
\end{equation}
In this case, no volume factors appear and we find 
$
K_{(bii)} = K_{(bi)}
$.
So the contributions to the propagator are directly identical here.

\begin{figure}[htp]
\begin{center}
   \begin{tikzpicture}  [scale=1.5]
  \begin{scope} [very thin]
    \draw [->]  (0,0) -- (-20:1.8)  ;
    \draw [->]  (0,0) -- (210:1.62)  ;
    \draw [->]  (0,0) -- (0,1.8)    ;
  \end{scope}
    \draw (0,1) ++(340:1) ++(210:0.9) -- +(0,-1) ;
    \draw (0,1) ++(340:1) ++(210:0.9) -- +(160:1) ;
    \draw (0,1) ++(340:1) ++(210:0.9) -- +(30:0.9) ;
    \filldraw [fill=black!20!white , draw =black, thick] 
         (0,1) ++(210:0.9) -- (0,1) -- ++(-20:1) -- cycle ;
    \filldraw [fill=black!20!white , draw =black, thick] 
         (210:0.9) ++(0,1) -- (210:0.9) -- ++(340:1) -- cycle  ;
    \filldraw [fill=black, thick] (210:0.9) circle (0.06)
          node[anchor = east] {$u_j \ $} ;
    \filldraw [fill=black, thick] (0,1)  circle (0.06) 
          node[anchor = south west] {$  u_k$} ;
    \filldraw [fill=black, thick] (0,1) ++(340:1) circle (0.06) 
          node[anchor=west] {$ \ u_{ki}$} ;    
    \filldraw [fill=black, thick] (210:0.9)++(340:1)  circle (0.06) 
          node[anchor = north] {$ u_{ij}$} ;
    \filldraw [fill=white, thick] (0,1) ++(210:0.9) circle (0.06)
          node[anchor = east] {$u_{jk} \ $} ;
    \filldraw [fill=white, thick] (0,1) ++(340:1) ++(210:0.9)
          circle (0.06) 
          node[anchor=north west] {$u_{ijk} $} ;
   \draw (-1.5,1.7)  node {(i)} ;
   \begin{scope} [xshift=100pt]
  \begin{scope} [very thin]
    \draw [->]  (0,0) -- (-20:1.8)  ;
    \draw [->]  (0,0) -- (210:1.62)  ;
    \draw [->]  (0,0) -- (0,1.8)    ;
  \end{scope}
    \filldraw [fill=black!20!white , draw =black, thick] 
         (0,1) -- (0,0) -- (340:1) -- cycle  ;
    \filldraw [fill=black!20!white , draw =black, thick] 
         (210:0.9) -- (0,0) -- (340:1) -- cycle  ;
    \filldraw [fill=black!20!white , draw =black, thick] 
         (0,1) -- (0,0) -- (210:0.9) -- cycle  ;
    \draw (-21:1) -- ++(0,1) -- ++(210:0.9) -- ++(0,-1) -- cycle;
    \draw (0,1) -- ++(-20:1) ;
    \draw (210:0.9) -- ++(-20:1) ;
    \draw (0,1) -- ++(210:0.9) -- ++(0,-1) ;
    \filldraw [fill=black!20!white , draw =black, thick] 
         (-20:1) ++(0,1) -- (-20:1) -- ++(210:0.9) -- cycle  ;
    \filldraw [fill=white, thick] circle (0.06) 
          node[anchor=east] {$u \ $} ;
    \filldraw [fill=black, thick] (210:0.9) circle (0.06)
          node[anchor = east] {$u_j \ $} ;
    \filldraw [fill=white, thick] (340:1) circle (0.06) 
          node[anchor=west] {$ \ u_i$} ;    
    \filldraw [fill=black, thick] (0,1)  circle (0.06) 
          node[anchor = south west] {$  u_k$} ;
    \filldraw [fill=black, thick] (0,1) ++(340:1) circle (0.06) 
          node[anchor=west] {$ \ u_{ki}$} ;    
    \filldraw [fill=black, thick] (210:0.9)++(340:1)  circle (0.06) 
          node[anchor = north] {$ u_{ij}$} ;
   \draw (-1.5,1.7)  node {(ii)} ;
    \end{scope}     
\end{tikzpicture}
\end{center}
\caption{Picture for elementary move (c)}
\label{fig:elementaryc}
\end{figure}
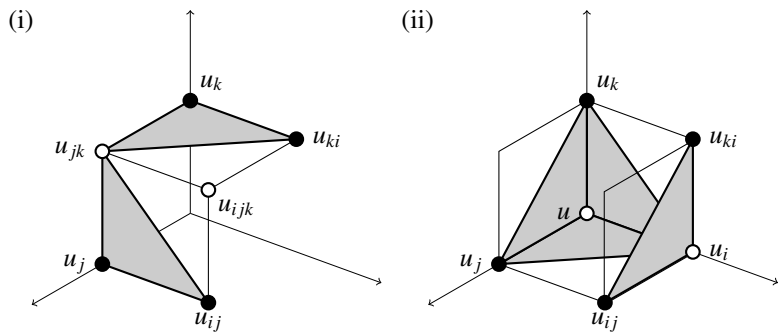

Lastly, consider elementary move (c) shown in figure \ref{fig:elementaryc}.
These bear a clear relation to figure \ref{fig:elementaryb}: the element
$\mcL_{jk} (u)$ has been shifted from one diagram to the other, inducing
also a slight change in the integration variables. 
For \ref{fig:elementaryc}(i):
\begin{equation}
    \mcS_{(ci)} = \mcL_{ij}(u_k) + \mcL_{ki} (u_j) 
\ , \quad
    K_{(ci)} = \mcN_{(ci)} \iint \ud u_{jk} \ud u_{ijk}
     \exp \left[ i \mcS_{(ci)} / \hbar \right] \ .
\end{equation}
Similarly, \ref{fig:elementaryc}(ii) is derived from \ref{fig:elementaryb}(i)
with an additional integral over $u$.
\begin{equation}
    \mcS_{(cii)} = \mcL_{ij}(u) + \mcL_{jk} (u)
                          + \mcL_{ki} (u) - \mcL_{jk} (u_i) 
\ , \quad
    K_{(cii)} = \mcN_{(cii)} \iint \ud u \ud u_{i}  \ 
     \exp \left[ i \mcS_{(cii)} / \hbar \right] \ .
\end{equation}
Once more we find that 
$
K_{(cii)} = K_{(ci)}
$
(although this time a volume factor is involved on both sides) and
the contributions to the propagator are the same.

\begin{prop}
The system characterised by Lagrangian (\ref{eq:qlatticelagrange})
is \emph{independent} of the choice of surface $\sigma$,
up to the choice of normalising constants.
\label{prop:latticeindependent}
\end{prop}

\begin{proof}
The combination of elementary moves above,
combined with the pop-up of figure \ref{fig:popup}, 
allows us to deform any surface $\sigma$ to another
topologically equivalent surface $\sigma'$ by a series of
elementary moves, without changing the exponent in the
propagator. This free deformation gives us independence from
the surface.
\end{proof}

An obvious consequence is that the propagator 
(\ref{eq:correlation}) depends only on the 
surface boundary $\partial \sigma$, and the field variables specified
there - i.e. it is a function only of the boundary value problem.
Note that since different topologies are
specified by changes of the boundary, we have not considered these explicitly.

\subsection{Uniqueness}
\label{sec:elementarygeneral}

The Lagrangian (\ref{eq:qlatticelagrange}) has the property
that it produces a propagator (\ref{eq:correlation}) which is
independent of variations of the surface $\sigma$.
In fact, it turns out that (\ref{eq:qlatticelagrange}) is the
unique quadratic Lagrangian 2-form such that this holds.
Consider a general, 3-point, quadratic
Lagrangian, imposing antisymmetry under interchange of $i$ and $j$:
\begin{equation}
    \mcL_{ij} (u, u_i, u_j) = \tfrac12 a_{ij} u^2 + \tfrac12 b_{ij} u_i^2
        - \tfrac12 b_{ji} u_j^2 + c_{ij} u u_i - c_{ji} u u_j
        + d_{ij} u_i u_j \ ,
\label{eq:generallatticelagrange}
\end{equation}
For coefficients, a subscript $i$ indicates dependence on the 
lattice parameter $p_i$, with the ordering of subscripts important.
The 2-form structure requires $a_{ji} = - a_{ij}$, $d_{ji} = - d_{ij}$
($a_{ij}$ and $d_{ij}$ are anti-symmetric under interchange of the parameters).
Our interest is in the subset of Lagrangians that display the
surface independence property in the propagator. We therefore look for conditions
on the Lagrangian such that elementary moves will leave the
contribution to the action (i.e. the exponent in the propagator)
unchanged.
We assume that extenal factors and even volume factors can be 
resolved by renormalisation, so that we only 
consider that part of the propagator in the exponent.


Consider (\ref{eq:generallatticelagrange}) under elementary move (a) - 
shown in figure \ref{fig:elementarya}. The contributions to the
propagator, $K_{(ai)}$ and $K_{(aii)}$, are calculated according to
(\ref{eq:Ka1}) and (\ref{eq:Ka2}). For surface independence, we require
$
K_{(ai)} = K_{(aii)}
$.

$K_{(ai)}$ is calculated via an integral $\ud u$, as in (\ref{eq:Ka1}).
In general, the coefficient of $u$ in the exponent may be either
quadratic, linear, or zero: yielding a Gaussian integral, Dirac delta function,
or volume factor, respectively. 
However, a Dirac delta function would force linear dependence of field variables at different lattice points: since this is undesirable, we exclude this
possibility.
The remaining cases divide on the totally antisymmetric
coefficient $\mca_{ijk}:= a_{ij} + a_{jk} + a_{ki}$ (see \ref{app:surfaceindependenceA} for 
details).
For $\mca_{ijk} \neq 0$ we have a Gaussian integral, and:
\begin{multline}
    K_{(ai),G} = 
         \left( \frac{2 \pi i \hbar}{\mca_{ijk}} \right)^{1/2}
         \exp \frac i{\hbar} \bigg[
                \frac12 \Big(b_{ij} - b_{ik}
                - \frac1{\mca_{ijk}} (c_{ij} - c_{ik})^2
                \Big) u_i^2 + cyclic
 \\ 
                + \Big( d_{ij} 
                - \frac1{\mca_{ijk}} (c_{ij} - c_{ik})(c_{jk} - c_{ji})
                \Big) u_i u_j + cyclic
                \bigg] \ .
\label{eq:elmoveiG}
\end{multline}
Conversely, for $\mca_{ijk} = 0$, we require the integral to reduce
to a volume factor (linear coefficients of $u$ in the exponent must
disappear) requiring the conditions
\begin{equation}
    a_{ij}  =  a_i - a_j
    \ , \quad
    c_{ij}  = c_i
    \ .
\label{eq:elcond1}
\end{equation}
(the coefficient $a_{ij}$ must separate into a part depending
on  $p_i$  and a part depending on  $p_j$ and
$c_{ij}$ is a function of $p_i$ only).
Under these conditions,
\begin{equation}
    K_{(ai),V} = V \ \exp \frac i{\hbar} \left[
                \tfrac12 (b_{ij} - b_{ik}) u_i^2 + cyclic
                + d_{ij} u_i u_j + cyclic
                \right] \ .
\label{eq:elmoveiV}
\end{equation}
This is a critical point of the variation - a volume factor appears
uniquely for this special choice of Lagrangian, which can be written
as:
\begin{eqnarray}
    \mcL_{ij} (u, u_i, u_j) &=&
    \tfrac12 a_i u^2 + c_i u u_i
    - \tfrac12 a_j u^2 - c_j u u_j
    + \tfrac12 (b_{ij} u_i^2 - b_{ji} u_j^2) + d_{ij} u_i u_j \ ,
\notag \\
    &=& A_i (u, u_i) - A_j (u, u_j) + C_{ij} (u_i, u_j) \ ,
\end{eqnarray}
with $C_{ij} (u_i, u_j)$ antisymmetric under interchange of $i$ and $j$.
This is the most general classical Lagrangian 2-form 
(\ref{eq:latticeEL})
as found in \cite{lobb2013variational}, here specialised to the
quadratic case.
So we have two cases for $K_{(ai)}$: (\ref{eq:elmoveiV}) when $\mca_{ijk}=0$,
and (\ref{eq:elmoveiG}) when $\mca_{ijk} \neq 0$.

For $K_{(aii)}$, as in (\ref{eq:Ka2}), we have four integrations
$\ud u_{ij} \ud u_{jk} \ud u_{ki} \ud u_{ijk}$. The integral
$\ud u_{ijk}$ always produces a volume factor due to the
three-point form of the Lagrangian. 
As for $K_{(ai)}$, we wish to avoid these integrals reducing
to a Dirac delta function, and so we have
2 cases. The remaining integrals are either evaluated as three Gaussian
integrations, or one integration reduces to a volume factor. This rests on
the value of $\det A$ (see \ref{app:surfaceindependenceB}
for details):
\begin{equation}
    A = 
        \left( \begin{array}{ccc}
                b_{jk} - b_{ik}  &  d_{ki}  &  d_{jk}   \\
                d_{ki}  &  b_{ki} - b_{ji}  &  d_{ij}  \\
                d_{jk}  &  d_{ij}  &  b_{ij} - b_{kj}
               \end{array}   \right)  \ .
\label{eq:matrixA}
\end{equation}
For $\det A \neq 0$ (equivalently $b_{ij} \neq - d_{ij}$)
we have three Gaussian integrations, producing:
\begin{equation}
    K_{(aii),G} = V \sqrt{ \frac{(2 \pi i \hbar)^3}{\det A}}
                  \exp \left( -\frac i{2 \hbar} \bB^T A^{-1} \bB \right)
                  \exp \frac i{\hbar} \bigg( 
                  \frac12 a_{jk} u_i^2 + cyclic
                  \bigg) \ , 
\label{eq:elmoveiiG}
\end{equation}
where
$
    \bB^T = \left( c_{jk} u_i  - c_{ik} u_j , 
    \textrm{perm } (ijk) , \textrm{perm } (kji) \right) \ 
$.
Alternatively, when $\det A = 0$, evaluating $K_{(aii)}$ requires
two Gaussian integrations. We then require linear terms in the third
integrand to disappear in order to prohibit the appearance of a 
Dirac delta function (see \ref{app:surfaceindependenceB})
hence
we require the conditions
\begin{equation}
    b_{ij} = - d_{ij} 
    \ , \quad
    c_{ij} = c_{ji}   \qquad \forall i,j \ .
\label{eq:elcond2}
\end{equation}
So, $b_{ij}$ is also anti-symmetric, and $c_{ij}$ symmetric.
We can then evaluate $K_{(aii)}$ as:
\begin{equation}
    K_{(aii),V} = \frac{2 \pi \hbar}{(1-\Lambda_{ijk})^{1/2}} V^2 \ 
                 \exp \frac i{\hbar} \bigg[
                 \frac12 a_{jk} u_i^2 + cyclic
                 - \frac12 \frac{d_{ij}}{1-\Lambda_{ijk}}
                 (c_{jk} u_i - c_{ki} u_j)^2 + cyclic
                 \bigg] \ , 
\label{eq:elmoveiiV}
\end{equation}
where we have introduced the totally symmetric parameter
\begin{equation}
    \Lambda_{ijk} := d_{ij} d_{jk} + d_{jk} d_{ki} 
                    + d_{ki} d_{ij} + 1 \ .
\label{eq:ellambda}
\end{equation}
Once more there are two cases. For $K_{(aii)}$, 
when $\det A = 0$, we find (\ref{eq:elmoveiiV}),
and when $\det A \neq 0$ we have (\ref{eq:elmoveiiG}).

Comparing now the two configurations of the elementary move,  we demand 
that the exponents from each configuration be the same;
i.e.  both make the same contribution to the propagator.
More details of this comparison are given in 
\ref{app:surfaceindependenceC}.
We find a solution to the problem
at the critical point of the system: where some of our integrals become singular.
Allowing $\mca_{ijk} =0$ and $\det A=0$,
we compare the exponent in (\ref{eq:elmoveiV}) with (\ref{eq:elmoveiiV}). 
Recalling that at this critical point we have also the conditions
(\ref{eq:elcond1}), (\ref{eq:elcond2}), we find that we require
$
    c_{ij} = c \ 
$, constant,
$
    \Lambda_{ijk} = 1 - c^2 \ 
$, $
    a_{ij} = 0 \ 
$.
Finally, since our Lagrangian is defined only up to an overall multiple,
we let $c=1$. 
We therefore find the unique quadratic Lagrangian:
\begin{equation}
    \mcL_{ij} (u, u_i, u_j) = u(u_i - u_j) - \tfrac12 d_{ij} (u_i - u_j)^2  \ ,
\label{eq:qlatticelagrange2}
\end{equation}
along with the condition on $d_{ij}$ that $\Lambda_{ijk}=0$. 
Comparing (\ref{eq:ellambda}) with (\ref{eq:sijidentity}) we see that we require
precisely $d_{ij} = s_{ij}$.
But then (\ref{eq:qlatticelagrange2}) is uniquely the Lagrangian 
(\ref{eq:qlatticelagrange})! 
We already know from section \ref{sec:elementaryspecial}
that this Lagrangian also exhibits surface independence for the other elementary moves.
This principle of surface independence is then sufficient to determine the 
required Lagrangian uniquely: even more so than in the classical case
(\ref{eq:linearlagrangian2}).

\begin{prop}
The Lagrangian (\ref{eq:qlatticelagrange}) is the unique quadratic
Lagrangian 2-form yielding a surface
independent propagator (\ref{eq:correlation}).
\end{prop}

\begin{pf}
(\ref{eq:qlatticelagrange2}), with the restriction $\Lambda_{ijk}=0$
(\ref{eq:ellambda}), gives  us that this
 is the unique Lagrangian exhibiting
surface independence for elementary move (a).
We also have from proposition \ref{prop:latticeindependent}
that Lagrangian (\ref{eq:qlatticelagrange}) has surface independence 
under all other elementary moves.
\end{pf}

\subsection{Quantum Variational Principle: Lagrangian 2-form case}

This result suggests a quantum variational principle in analogy to the one
dimensional case of section \ref{sec:qvariation1d}.
We consider the propagator over a discrete surface $\sigma$,
$K_{\sigma} (\partial \sigma)$,
defined in (\ref{eq:correlation}). 
We have shown that, for the special choice of Lagrangian (\ref{eq:qlatticelagrange}),
the propagator $K_{\sigma} (\partial \sigma)$ is \emph{independent}
of the surface $\sigma$. It depends only on the variables sitting
on the boundary, $\partial \sigma$.
Additionally, this is a very unique choice of Lagrangian:
for a generic Lagrangian, $K_{\sigma} (\partial \sigma)$
will depend also on the surface $\sigma$ itself.

Recall that, classically, the Lagrangian 2-form structure arises
from a variational principle \emph{over surfaces} as in \cite{lobb2013variational}. 
An extended set of Euler-Lagrange equations arise as we vary not only the 
dependent field variables $u_{\bn}$, but also the surface $\sigma$.
This restricts the class of admissible Lagrangians to those
obeying the closure property (\ref{eq:closure}): it is only for
such Lagrangians and equations of motion that the classical action
remains invariant under variations of the surface.

As we move to the quantisation, parallel to what we argued in the 1-form case, we consider the variation 
over all possible surfaces $\sigma$ with a fixed boundary $\partial \sigma$.
For a generic Lagrangian, as we vary the surface $\sigma$ the 
propagator $K_{\sigma} (\partial \sigma)$ \eqref{eq:correlation} changes.
However, for the special ``integrable'' choice of Lagrangians \eqref{eq:qlatticelagrange}
the propagator $K_{\sigma} (\partial \sigma)$ remains unchanged 
as we vary the surface. This therefore represents a critical (i.e., singular) point 
for a new quantum object which we conjecture to be a ``sum over all surfaces'' of which the 
surface-dependent propagator forms the summand\footnote{The sum over surfaces idea has also emerged in the 
theory of loop quantum gravity but with a different motivation, cf. \cite{reisenberger1997sum,reisenberger1997left}.}, viewed as a functional in a well-chosen 
space of Lagrange functions. Once again, controlling the singular behaviour of such an object, 
and arriving at mathematically concise definition is the subject of ongoing investigation. 
Nonetheless, we conjecture that critical/singular point analysis of such an object, leading to the 
selection of Lagrangians whose propagator are surface-independent, would form a key ingredient for 
understanding the path integral quantisation of discrete field theories that are integrable in the 
sense of multidimensional consistency.


\section{Discussion}

In his seminal paper of 1933, \cite{dirac1933lagrangian}, Paul Dirac expressed his \textit{credo} that the Lagrangian formulation of 
classical dynamics, in comparison to the Hamiltonian one, was more fundamental, and he posed the question of a Lagrangian approach to quantum mechanics. In this important precursor to Feynman's development of the path integral \cite{feynman1948space} 
the analogy between classical and quantum mechanics was emphasized, cf. also \cite{dirac1945analogy}. In this context, the related question of 
what would constitute a variational point of view in quantum mechanics was partly, but not fully, answered by those approaches.   
In the present paper we have attempted to arrive to a more complete answer to these questions in the context of integrable systems in 
the sense of multidimensional consistency. This is pursued by setting up a quantum analogue of the Lagrangian multi-form approach.  

The main result, obtained within the context of quadratic Lagrangians, is that there is a quantum analogue 
of the closure property of \cite{lobb2009lagrangian} which 
underlies the classical multiform theory. The quantum analogue is formulated in terms of the multi-time 
propagators for these models, cf. eq. (\ref{eq:gammapropagator}). 

There are a number of points to make in connection with the results obtained in this study. 

First, although the results were obtained by restricting ourselves to only quadratic Lagrangians, 
the multidimensional consistency aspects do not essentially rely on the linearity of the equations. In fact, most of the combinatorics at the classical level  
carries through for all Lagrangians associated with nonlinear quad equations in the ABS list, cf. \cite{xenitidis2011lagrangian}. Due to the suspected close 
analogy between classical theory and quantum theory in the integrable case, it is therefore to be expected that some quantization procedure for those 
models would exist such that the results obtained here also carry through to the quantum level for those nonlinear models. This may, however, require 
non-conventional quantization prescriptions in terms of suitable integrals replacing the Gaussian integrals used in the quadratic case. Initial results 
along this direction were obtained in \cite{field2005thesis} and \cite{field2006time}. The choice of Hilbert space (in the canonical 
quantization picture), and of integration measure (in the path integral picture) may be driven by the integrable combinatorics of those models.    

Second, another general feature of the models in question is the role-reversal interplay between parameters and independent variables and between the 
discrete and continuous models. Thus, the continuous models do not only appear as continuum limits, but more intrinsically as additional commuting 
flows: the classical equations hold simultaneously on a common set of solutions. On the quantum level this property extends in the fact that there is a 
common propagator of the underlying continuous and discrete quadratic models. If this feature is general enough to extend to the nonlinear case 
(which it does in the classical case) there is scope that this property can eventually be used to extract information on the time-sliced path integral 
from the discrete finite-step path integral. 

Third, turning things around and imposing the path and surface independence of the propagator for a general parameter class of quadratic Lagrangians, we have 
shown that this quantum MDC property leads uniquely to the Lagrangians that arise from the integrable case, in the same spirit as in \cite{lobb2013variational}. 
In fact, the point made in that paper is that the Lagrangians themselves should be viewed as solutions of an extended set of Euler-Lagrange equations, 
which incorporates the stationarity under variations with respect to both the field (i.e., dependent) variables as well as the geometry in the independent variables. 
This poses a new paradigm in variational calculus, as it signifies a departure from the conventional point of view of most physical theories, namely 
that Lagrangians have to be chosen based on tertiary considerations. In this new point of view, the Lagrangians are not necessarily given in advance, but follow from the 
variational principle itself.    

We finish by making a few general remarks on further ramifications. In general it is not known how to derive a path integral formalism for non-conventional, i.e. 
non-Newtonian models, through a time-slicing procedure when Gaussian integrals no longer apply. Nonetheless, in integrable systems 
theories such non-Newtonian models do abundantly appear and often can also be readily quantized through the canonical formalism, e.g. the relativistic many-body systems of 
Ruijsenaars-Schneider type, \cite{ruijsenaars1986new}. This poses, in our view, a lacuna in the theory which is imperative to rectify as such integrable quantum systems 
cannot be simply discarded as potentially physical models. Thus, integrable systems can play a role of a litmus test for the completeness of a theory, which most reasonably should 
be applicable to those models for which in principle exact and rigorous computations can be performed. However, one may speculate that there is a deeper significance 
for those systems, since they have proved their merit in forming a fruitful breeding ground for new concepts and new understandings on a fundamental level. 
In fact, the ideas exposed in the present paper, based on simple toy prolems, have some interesting resemblances to proposals that 
that in recent years have been put forward on the quantization of scaling invariant theories \cite{rovelli2011discretizing,rovelli2011structure,bahr2011perfect}.  
A particular parallel may be drawn between path and surface independence of propagators in our examples, and certain formulations of loop quantum 
gravity and "sum over surfaces", \cite{reisenberger1997sum,reisenberger1997left}. Furthermore, the interplay between discrete and continuous, which is prominent 
in our examples, may perhaps feed into views that G.'t Hooft has been promoting with regard to the quantum nature of the universe, cf. \cite{thooft2014cellular}.

\section*{Acknowledgments}

Steven King would like to thank EPSRC for funding this research.
Frank Nijhoff was partly supported by EPSRC grant EP/I038683/1. 
FWN acknowledges early discussions with Sikarin Yoo-Kong and Chris Field on related 
topics.


\appendix

\section{Calculating the Discrete Propagator}
\label{app:discretepropagator}

\subsection{The Classical Action}

We wish to evaluate the classical action for the path beginning at $x_0$, and reaching
$x_N$ after $N$ time steps. Recalling our discrete equation of motion 
(\ref{eq:p=1.5xb}) and classical solution 
we have the classical path
\begin{equation}
     x_n = 
     \frac1{\sin \mu N} \big(x_N \sin \mu N - x_0 \sin \mu (n-N) \big) \ .
\end{equation}
The action along the classical path is then:
\begin{equation}
    \mathcal S_{cl} = \sum_{n=0}^{N-1} \left( \frac{P+Q}q x_n x_{n+1}
                 + \frac{P-Q}{2q} (x_n^2 + x_{n+1}^2)   \right)  
    =
    \frac{\sqrt P}{\sin \mu N}
         \left[ 2x_0 x_N - (x_0^2 + x_N^2) \cos \mu N \right] \ ,
\end{equation}
where we have used the identities:
\begin{equation}
    \cos \mu = -b = - \frac{P-Q}{P+Q} 
\ , \quad
    \sin \mu = \frac{2q \sqrt P}{P+Q} \ .
\end{equation}
We note two things about this result. First, there is no explicit $Q$ dependence: 
all $Q$ dependence is contained within the parameter $\mu$, which only appears
as $\mu N$. Second, we can easily extend this result to the $\mathcal L_a$ (bar evolution)
case, by a change of parameter. We replace $\mu$ by $\eta$, such that
$\cos \eta = -a$.

\subsection{The discrete propagator}

It is left for us to evaluate the discrete path integral:
\begin{equation}
    \tilde K_N (0, 0) = \int_{y(0)=0}^{y(N)=0} \mathcal D[y_n] 
                               \ e^{i \mathcal S[y_n]/\hbar} \ .
\end{equation}
In the discrete case, we can consider this via a time slicing procedure without
needing to worry about the problematic shrinking to zero. So we consider:
\begin{equation}
    \tilde K_N(0,0) = \mathcal N \int \ud y_1 \ldots \int \ud y_{N-1}  
         \   \exp \left\{ \frac i{\hbar q} 
            \sum_{n=0}^{N-1} \left( (P+Q) y_n y_{n+1}
             + \frac12 (P-Q) (y_n^2 + y_{n+1}^2)
              \right)  \right\}  \ ,
\end{equation}
where $\mcN$ is the normalising factor appearing in (\ref{eq:discretepathintegral})
and $y_0 = y_N = 0$.
This expression is quadratic in all $y_n$ variables, and so can be evaluated
as $N-1$ Gaussian integrals. This is most easily achieved by writing the
equation in a matrix form (as in \cite{schulman2005techniques}, for example).
We define $ \mathbf y^T = (y_1, \ldots, y_{N-1})$, in order to write
\begin{equation}
    \tilde K_N = \mathcal N \int \ud^{N-1} \mathbf y \ 
                       \exp (- \mathbf y^T \sigma \mathbf y) 
\ 
    =  \frac{\pi^{(N-1)/2}}{\sqrt{\det \sigma}} \ ,
\end{equation}
with $\sigma$ the symmetric, tri-diagonal matrix:
\begin{equation}
    \sigma = \frac{i(P+Q)}{\hbar q}
        \begin{pmatrix}
             -\frac{P-Q}{P+Q}  &  -1/2   &  &   \\
             -1/2   &  -\frac{P-Q}{P+Q}  & \ddots  &   \\
             &  \ddots  & \ddots  & -1/2  \\ 
             &  &  -1/2   & -\frac{P-Q}{P+Q} 
        \end{pmatrix} \ .
\end{equation}
Hence it remains to calculate $\det \sigma$.
The determinant for a tri-diagonal matrix can be found by forming a recursion
relation on the size of the matrix, and solving as a discrete equation.
Let
\begin{equation}
    X_n =  \begin{vmatrix}
                 a & b & & \\
                 b & a & \ddots & \\
                 & \ddots & \ddots & b \\
                 & & b & a 
              \end{vmatrix} \quad
      \textrm{of size} \ n \ .
\end{equation}
Performing the cofactor expansion, we find
\begin{equation}
     X_n = a X_{n-1} - b^2 X_{N-2} \ , 
\end{equation}
with initial conditions $X_1 = a$ and $X_2 = a^2 - b^2$. The solution is 
thus given by
\begin{equation}
    X_n = \frac1{2^n} \Bigg[ 
        \left( a + \frac{a^2 - 2b^2}{\sqrt{a^2 - 4b^2}} \right)
        \left( a + \sqrt{a^2 - 4b^2} \right)^{n-1}  
        + \left( a - \frac{a^2 - 2b^2}{\sqrt{a^2 - 4b^2}} \right)
        \left( a - \sqrt{a^2 - 4b^2} \right)^{n-1}
      \Bigg] \ . 
\end{equation}
Now, in the case of $\sigma$, recall that $a = -(P-Q)/(P+Q) = \cos \mu$
and $b = -1/2$, so that $\sqrt{a^2 - 4b^2} = i \sin \mu$: this leads to
significant simplifications of the above expression. Working through
these calculations, we then find:
\begin{equation}
    \det \sigma = \left( \frac{i(P+Q)}{2 \hbar q} \right)^{N-1}
                          \frac{\sin \mu N}{\sin \mu} \ . 
\end{equation}

Putting this together, then,
\begin{equation}
    \tilde K_N = 
    \left( \frac{i(P+Q)}{2 \pi \hbar q} \right)^{N/2}
      \left( \frac{2 \pi \hbar q}{i(P+Q)} \right)^{(N-1)/2}
                          \sqrt{\frac{\sin \mu}{\sin \mu N}} \ ,
\end{equation}
and therefore
\begin{equation}
    K_N (x_0, x_N) =  
    \left( \frac{i \sqrt P}{\pi \hbar \sin \mu N} \right)^{1/2}
                \       \exp \left\{
                          \frac{i \sqrt P}{\hbar \sin \mu N}
         \left( 2x_0 x_N - (x_0^2 + x_N^2) \cos \mu N \right) \right\} \ . 
\end{equation}

\section{Quantum Invariants}
\label{app:quantuminvariants}

In \cite{field2006time}, the authors investigated quantum systems possessing 
invariants under a one time-step path integral evolution.
Begin by considering the evolution in the hat direction,
generated by $\mathcal L_b (x, \oh x)$ (\ref{eq:p=1.5lagrangians}).
A wavefunction $\psi_n (x)$ evolves under this tranformation according to
\begin{equation}
    \psi_{n+1} (\oh x) = \mathcal N \int_C \exp \left( \frac i{\hbar}
                          \mathcal L_b (x, \oh x) \right) \psi_n (x) \ud x \ ,
\end{equation}
and to look for an invariant we desire $\psi_n$ and $\psi_{n+1}$ to be 
solutions of the same eigenvalue problem, with the same eigenvalue:
\begin{equation}
    M_x \psi_n (x) = E \psi_n (x) \quad
    \Rightarrow \quad
    M_{\oh x} \psi_{n+1} (\oh x) = E \psi_{n+1} (\oh x) \ .
\end{equation}
$M_x$ is a differential operator, and we restrict to considering the second order
case:
\begin{equation}
    M_x = p_0 (x) \frac{\partial^2}{\partial x^2}
                + p_1 (x) \frac{\partial}{\partial x} + p_2(x) \ .
\end{equation}
Now, 
\begin{equation}
    E \psi_{n+1} (\oh x) =
        \mathcal N \int_C \exp \left( \frac i{\hbar}
        \mathcal L_b (x, \oh x) \right) \left( M_x \psi_n (x) \right) \ud x 
\  
    =
    \mathcal N \int_C \left( \ob M_x \exp \left( \frac i{\hbar}
               \mathcal L_b (x, \oh x) \right) \right) \psi_n (x) \ud x \  
                + \mathcal S \ ,
\end{equation}
where $\ob M_x$ is an adjoint to $M_x$ constructed under integrations by parts,
and $\mathcal S$ is the resulting surface term.
If we assume $\psi_n$ and $\psi_n'$ to vanish at infinity (a reasonable physical assumption)
then the surface term $\mathcal S$ vanishes.
We can also write,
\begin{equation}
    E \psi_{n+1} (\oh x) =  
    M_{\oh x} \psi_{n+1} (\oh x) 
\ 
    =
    \mathcal N \int_C \left(  M_{\oh x} \exp \left( \frac i{\hbar}
        \mathcal L_b (x, \oh x) \right) \right)  \psi_n (x) \ \ud x \ .
\end{equation}
So the condition we require is for 
$\ob M_x \exp \left( \frac i{\hbar} \mathcal L_b (x, \oh x) \right) 
= M_{\oh x} \exp \left( \frac i{\hbar} \mathcal L_b (x, \oh x) \right) \ $. 
Following the analysis in \cite{field2006time}, and using the given Lagrangian,
we find this can only hold under the restrictions:
\begin{equation}
    p_0 (x) = - \hbar^2 C_0 \ , \quad
    p_1 (x) \equiv 0 \ , \quad
    p_2 (x) = 4P C_0 x^2 + C_2 \ ,
\end{equation}
so that
\begin{equation}
    M_x = C_0 \left(- \hbar^2 \frac{\partial^2}{\partial x^2} 
                 + 4P x^2 \right) + C_2 \ .
\end{equation}
This is precisely the quantum invariant (\ref{eq:qinvariant}).

\section{Path Independence for a General Lagrangian}
\label{app:pathindepgeneral}

We calculate the propagators (\ref{eq:Klrcorner}) and (\ref{eq:Kulcorner})
by a Gaussian integral:
\begin{multline}
    K_{\lrcorner} (x, \oh {\ob x}) =
    \mcN_{\lrcorner}
    \left( \frac{\pi i \hbar}{\beta b_0 + \alpha (a-a_0)} \right)^{1/2}
          \exp \Bigg\{ \frac i{\hbar} \Bigg[  
            \left( \beta (b-b_0) - \frac{\beta^2}
                       {4 (\beta b_0 + \alpha (a-a_0))} \right) x^2
\\
            + \left( \alpha a_0 - \frac{\alpha^2}
                         {4 (\beta b_0 + \alpha (a-a_0))} \oh {\ob x}^2  \right)
            - \frac{\alpha \beta}{2 (\beta b_0 + \alpha (a-a_0))} x \oh {\ob x}
                    \Bigg] \Bigg\} \ ,
\label{eq:Klrgen}
\end{multline}
and,
\begin{multline}
    K_{\ulcorner} (x, \oh {\ob x}) = 
    \mcN_{\ulcorner}
    \left( \frac{\pi i \hbar}{\alpha a_0 + \beta (b-b_0)} \right)^{1/2}
          \exp \Bigg\{ \frac i{\hbar} \Bigg[  
            \left( \alpha (a-a_0) - \frac{\alpha^2}
                       {4 (\alpha a_0 + \beta (b-b_0))} \right) x^2
\\
           + \left( \beta b_0 - \frac{\beta^2}
                         {4 (\alpha a_0 + \beta (b-b_0))} \oh {\ob x}^2  \right)
            - \frac{\alpha \beta}{2 (\alpha a_0 + \beta (b-b_0))} x \oh {\ob x}
                    \Bigg] \Bigg\} \ .
\label{eq:Kulgen}
\end{multline}
By comparing the coefficients of $x^2$, $\oh{\ob x}^2$ and $x \oh {\ob x}$ in the 
exponent, we derive conditions for time-path-independence on our coefficients:
\begin{equation}
    \beta (b-b_0) - \frac{\beta^2}{4 (\beta b_0 + \alpha (a-a_0))}
      =
    \alpha (a-a_0) - \frac{\alpha^2}{4 (\alpha a_0 + \beta (b-b_0))} \ , 
\end{equation}
\begin{equation}
    \alpha a_0 - \frac{\alpha^2}{4 (\beta b_0 + \alpha (a-a_0))}
      =
    \beta b_0 - \frac{\beta^2}{4 (\alpha a_0 + \beta (b-b_0))} \ , 
\end{equation}
\begin{equation}
    \frac{\alpha \beta}{2 (\beta b_0 + \alpha (a-a_0))}
      =
    \frac{\alpha \beta}{2 (\alpha a_0 + \beta (b-b_0))}
\label{eq:qcompcond3} \ .
\end{equation}
Note that an immediate consequence of (\ref{eq:qcompcond3}) is that the multiplicative
factors in (\ref{eq:Klrgen}) and (\ref{eq:Kulgen}) are the same. Analysis of
these three conditions leads to (\ref{eq:pathindepconditions}).

\section{Elementary Moves}
\label{app:elmovea}

We consider elementary move (a), shown in figure \ref{fig:elementarya},
in more detail as an illustrative case.
The action and contributions to the propagator for figures
\ref{fig:elementarya}(i) and (ii) are given in
(\ref{eq:Ka1}) and (\ref{eq:Ka2}).
We then have
\begin{eqnarray}
    K_{(ai)} &=& \int \ud u \exp \frac i{\hbar} \bigg(
                   u (u_i - u_j) - \tfrac12 s_{ij} (u_i - u_j)^2
   \notag   \\   && \ \ 
                   + u (u_j - u_k) - \tfrac12 s_{jk} (u_j - u_k)^2
                   + u (u_k - u_i) - \tfrac12 s_{ki} (u_k - u_i)^2
                      \bigg) \ , 
    \notag   \\  
    &= & V \exp \frac{-i}{2 \hbar}
             \bigg( s_{ij} (u_i - u_j)^2 + s_{jk} (u_j - u_k)^2
                   +  s_{ki} (u_k - u_i)^2
                      \bigg)  \ ,
\label{eq:Kai}
\end{eqnarray}
where we note that all the $u$ terms have cancelled out, leaving
a volume factor.
We compare this to
\begin{eqnarray}
    K_{(aii)} &=& \iiiint \ud u_{ij} \ud u_{jk} \ud u_{ki} \ud_{ijk}
     \exp \frac i{\hbar} \bigg(
                   u_k (u_{ki} - u_{jk}) - \tfrac12 s_{ij} (u_{ki} - u_{jk})^2
   \notag   \\   && \ \ 
                   + u_i (u_{ij} - u_{ki}) - \tfrac12 s_{jk} (u_{ij} - u_{ki})^2
                   + u_j (u_{jk} - u_{ij}) - \tfrac12 s_{ki} (u_{jk} - u_{ij})^2
                      \bigg) \ , 
    \notag \\
    &= & V \int \ud^3 \bu \
     \exp \frac i{\hbar} \left(- \frac12 \bu^T A \bu + \bB^t \bu
       \right) \ , 
\end{eqnarray}
where
\begin{eqnarray}
    \bu^T &=& (u_{ij}, u_{jk}, u_{ki}) \ ,
\notag  \\
    A &=& 
        \left( \begin{array}{ccc}
                s_{jk} + s_{ki}  &  -s_{ki}  &  -s_{jk}   \\
                -s_{ki}  &  s_{ki} + s_{ij}  &  -s_{ij}  \\
                -s_{jk}  &  -s_{ij}  &  s_{ij} + s_{jk})
               \end{array}   \right)  \ ,
\\
    \bB^T &=& \left( u_i  - u_j , u_j - u_k , u_k - u_i \right) \ .
\notag
\end{eqnarray}
Critically, we note that $\det A = 0$, so again we have a singular integral.
Carrying out two integrals in turn, so that the third integration produces a
volume factor, we therefore have:
\begin{equation}
    K_{(aii)} =  V^2 2 \pi \hbar  \ 
     \exp \frac{-i}{2 \hbar}
             \bigg( s_{ij} (u_i - u_j)^2 + s_{jk} (u_j - u_k)^2
                   +  s_{ki} (u_k - u_i)^2
                      \bigg)  \ .
\end{equation}
Thus, the \emph{exponents} in $K_{(ai)}$ and $K_{(aii)}$ are the same.
With the correct choice of normalisation and regularisation, we have
identical contributions to the propagator.

\section{Uniqueness of the Surface Independent Lagrangian}
\label{app:surfaceindependence}

\subsection{Elementary move (a), configuration (i)}
\label{app:surfaceindependenceA}

For Lagrangian (\ref{eq:generallatticelagrange}), the expression
for $K_{(ai)}$ is shown in figure \ref{fig:elementarya} and
given by (\ref{eq:Ka1}). We then have:
\begin{multline}
    K_{(ai)}
                = \int \ud u \ \exp \frac i{\hbar} \Big[
                \tfrac12 (a_{ij} + a_{jk} + a_{ki}) u^2
                + \big( (c_{ij} - c_{ik}) u_i + (c_{jk} - c_{ji}) u_j
                       + (c_{ki} - c_{kj}) u_k  \big)  u  \Big]
\\ \times
                \exp \frac i{\hbar} \Big[
                \tfrac12 (b_{ij} - b_{ik}) u_i^2 + cyclic
                + d_{ij} u_i u_j + cyclic
                \Big] \ .
\label{eq:generalpopa}
\end{multline}
This integral is Gaussian providing the coefficient of $u^2$ does
not vanish; i.e. $\mca_{ijk} \neq 0$. In that case the integral
yields (\ref{eq:elmoveiG}).

The other case occurs when 
$
\mca_{ijk} = 0 
\ \Rightarrow \ 
a_{ij}  =  a_i - a_j
$.
 To avoid the integral
producing a delta function (which would threaten the independence
of our field variables) we then also require terms linear in $u$ to
vanish, so that 
$
c_{ij} - c_{ik}  =  0 
      \ \forall i,j,k
\ 
\Rightarrow \ 
      c_{ij}  =  c_i
$.
In other words, $c_{ij}$ must be a function of $p_i$ only.
These are precisely the conditions (\ref{eq:elcond1}).
If these conditions hold, we are left with the contribution to
the propagator (\ref{eq:elmoveiV}).

 \subsection{Elementary Move (a), configuration (ii)}
 \label{app:surfaceindependenceB}

For $K_{(aii)}$ in figure \ref{fig:elementarya}(ii), we have a
contribution to the propagator given by (\ref{eq:Ka2}). For
Lagrangian (\ref{eq:generallatticelagrange}) this gives us:
\begin{equation}
    K_{(aii)} 
           =  V \iiint \ud^3 \bu \
                \exp \frac i{\hbar} \bigg(\frac12 \bu^T A \bu 
                + \bB^t \bu \bigg)
                \exp \frac i{\hbar}
                \bigg[ \frac12 (a_{jk} u_i^2 + cyclic)
                \bigg] \ , 
\label{eq:generalpopb}
\end{equation}
with $A$ and $\bB$ as in (\ref{eq:matrixA}) and (\ref{eq:elmoveiiG}),
and
$
    \bu^T = (u_{ij}, u_{jk}, u_{ki})
$.
Clearly, when $\det A \neq 0$ this can be evaluated as a trio of
Gaussian integrals, giving (\ref{eq:elmoveiiG}). We must consider
the critical point $\det A = 0$ separately.
The condition $\det A=0$ is a \emph{functional equation} connecting
the $b_{ij}$ with the $d_{ij}$. Considering the rows of $A$ in 
(\ref{eq:matrixA}), it is clear that $\det A=0$ if
$
    b_{ij} = - d_{ij}
$,
our first condition of (\ref{eq:elcond2}). In this case we must carry out 
the two remaining Gaussian integrals in turn.
First, integrating over $\ud u_{ij}$
in (\ref{eq:generalpopb}):
\begin{multline}
    K_{(aii),V} = V \left( \frac{2 \pi \hbar}{i(d_{jk} + d_{ki})} \right)^{1/2}
                \iint \ud u_{jk} \ud u_{ki} \ 
                \exp \frac i{\hbar} \bigg[
                \frac12 \frac{1-\Lambda_{ijk}}{d_{jk} + d_{ki}}
                (u_{jk} - u_{ki})^2
\\
                + \Big( c_{ki} u_j - c_{ji} u_k 
                + \frac{d_{ki}}{d_{jk} + d_{ki}}
                (c_{jk} u_i - c_{ik} u_j) \Big)  
                (u_{jk} - u_{ki})
\\
                + \Big( 
                (c_{jk} - c_{kj}) u_i 
                + (c_{ki} - c_{ik}) u_j
                + (c_{ij} - c_{ji}) u_k
                \Big)  u_{ki}
                \bigg]
\\
                \exp \frac i{\hbar} \bigg[
                \frac12 \Big(
                a_{jk} u_i^2 + cyclic
                + \frac1{d_{jk} + d_{ki}} (c_{jk} u_i - c_{ik} u_j)
                \Big)  \bigg]        \ ,
\end{multline}
with $\Lambda_{ijk}$ given in (\ref{eq:ellambda}).
Here it is clear that we can shift our integration
by the subsitution $v = u_{jk} - u_{ki}$. 
Thus, to avoid a delta function integral for
$\ud u_{ki}$ and gain the volume factor we desire, we require also all 
terms linear in $u_{ki}$ in the exponent to vanish. Hence we require:
$
     c_{ij} - c_{ji} = 0  \ \forall i,j 
$.
This is the second condition of (\ref{eq:elcond2}).
Evaluation of the second Gaussian integral then gives us
(\ref{eq:elmoveiiV}).

\subsection{Elementary move (a): comparing results}
\label{app:surfaceindependenceC}

In the generic case ($\mca_{ijk} \neq 0$, $\det A \neq 0$)
we compare equation (\ref{eq:elmoveiG}) with (\ref{eq:elmoveiiG}). Comparing 
coefficients of $u_i^2$ and $u_i u_j$ in the exponent, this 
gives the functional equations:
\begin{multline}
    b_{ij} - b_{ik}
    - \frac1{\mca_{ijk}} (c_{ij} - c_{ik})^2
    =
    a_{jk} + \frac1{\det A} \Big\{ \big(
    d_{ij}^2 - (b_{ki} - b_{ji})(b_{ij} - b_{kj})
    \big) c_{jk}^2
\\
    + \big(
    d_{ki}^2 - (b_{jk} - b_{ik})(b_{ki} - b_{ji})
    \big) c_{kj}^2
    + 2 \big(
    d_{ij} d_{ki} - d_{jk} (b_{ki} - b_{ji})
    \big) c_{jk} c_{kj}
    \Big\}    \ ,
\label{eq:appgenerallattice1}
\end{multline}
\begin{multline}
    d_{ij} 
    - \frac1{\mca_{ijk}} (c_{ij} - c_{ik})(c_{jk} - c_{ji})
\\
    =
    \frac1{\det A} \bigg[
    \big((b_{ki} - b_{ji})(b_{ij} - b_{kj}) 
    - d_{ij}^2 \big) c_{jk} c_{ik}
    - \big(d_{ij} d_{jk} - d_{ki} (b_{ij} - b_{kj})
    \big) c_{jk} c_{ki}
\\
    + \big(d_{jk} d_{ki} - d_{ij} (b_{jk} - b_{ik})
    \big) c_{ki} c_{kj}
    - \big(d_{ij} d_{ki} - d_{jk} (b_{ki} - b_{ji})
    \big) c_{ik} c_{kj}
    \bigg]   \ .
\label{eq:appgenerallattice2}
\end{multline}
It is not at all obvious that a solution to these equations, under the constraints, exists.

However, in the special case $\mca_{ijk} = 0$, $\det A = 0$
we compare the exponent in (\ref{eq:elmoveiV}) with (\ref{eq:elmoveiiV}).
This gives equations from the 
coefficients of $u_i^2$ and $u_i u_j$:
\begin{equation}
    b_{ij} - b_{ik}
    =
    a_{jk}
    - \frac1{1-\Lambda_{ijk}}
    (d_{ij} + d_{ki}) c_{jk}^2 
\ , \quad
    d_{ij}
    =
    \frac{d_{ij}}{1-\Lambda_{ijk}}
    c_{jk} c_{ki}
\ . 
\end{equation}
Combined with the constraints (\ref{eq:elcond1}), (\ref{eq:elcond2}),
this yields the Lagrangian (\ref{eq:qlatticelagrange2}).

\bibliographystyle{unsrt}
\bibliography{linearkdvrefs}
\end{document}